\documentclass[pra,aps,notitlepage,nofootinbib]{revtex4-1}

\usepackage{bm}
\usepackage{bbm}
\usepackage{epsfig}
\usepackage{wrapfig}
\usepackage{float}
\usepackage{setspace}
\usepackage{mathrsfs}
\usepackage[normalem]{ulem}
\usepackage{enumerate}
\usepackage{tikz}
\usepackage{amsmath,amssymb,amsfonts,amsthm}
\usepackage{graphicx,xcolor}
\usepackage[a4paper,margin=2.4cm]{geometry}

\definecolor{mygenta}{cmyk}{0.25, 1, 0.25, 0.25}

\newcommand{\ignore}[1]{}
\newtheorem{theorem}{Theorem}
\newtheorem{proposition}[theorem]{Proposition}
\newtheorem{corollary}[theorem]{Corollary}
\newtheorem{maintheorem}[theorem]{Theorem}
\newtheorem{lemma}[theorem]{Lemma}

\newcommand{\surfaceangle}{%
\begin{tikzpicture}%
\draw (1.5ex,1.5ex) to [out=215,in=50] (0,0) to [out=-15,in=180]
(1.4ex,-0.1ex);%
\end{tikzpicture}%
}
\ignore{

}
\newcommand{\sphericaltriangle}{%
\begin{tikzpicture}%
\draw (0,0) to [out=70,in=230] (0.86ex,1.5ex) to [out=320,in=110] (1.72ex,0) to [out=190,in=-10] (0,0);%
\draw (0,0) to [out=70,in=230] (0.86ex,1.5ex) to [out=320,in=110] (1.72ex,0);%
\end{tikzpicture}%
}
\newcommand{\continuousp}[2]{\ensuremath{p_{#1}(#2)}}

\newcommand{\br}[1]{{\stackrel{\curvearrowright}{#1}}}
\newcommand{\lo}{\theta}
\newcommand{\la}{\delta}
\newcommand{\jp}{\phi}
\newcommand{\cphase}{\lo}
\newcommand{\cjump}{\jp}

\begin{document}

\title{Globe-hopping}

\author{Dmitry \surname{Chistikov}}
\affiliation{Centre for Discrete Mathematics and its Applications
  (DIMAP) \& Department of Computer Science, University of Warwick,
Coventry\hspace{0mm} CV4 7AL, U.K.} 
\email{d.chistikov@warwick.ac.uk}

\author{Olga \surname{Goulko}}
\affiliation{Boise State University, Department of Physics, 1910 University Drive
Boise, ID 83725-1570, U.S.A.}
\email{olgagoulko@boisestate.edu}

\author{Adrian \surname{Kent}}
\affiliation{Centre for Quantum Information and Foundations, DAMTP, Centre for
  Mathematical Sciences, University of Cambridge, Wilberforce Road,
  Cambridge CB3 0WA, U.K.}
\affiliation{Perimeter Institute for Theoretical Physics, 31 Caroline Street North, Waterloo, ON N2L 2Y5, Canada.}
\email{apak@damtp.cam.ac.uk} 

\author{Mike \surname{Paterson}}
\affiliation{Centre for Discrete Mathematics and its Applications (DIMAP) \& Department of Computer Science, University of Warwick, Coventry CV4 7AL, U.K.}
\email{m.s.paterson@warwick.ac.uk}

\date{\today} 

\begin{abstract}
We consider versions of the grasshopper problem~\cite{goulko2017grasshopper} on the circle and the sphere, which are relevant to Bell inequalities.  
For a circle of circumference $2\pi$, 
we show that for unconstrained lawns of any length and arbitrary jump lengths, the supremum of the probability 
for the grasshopper's jump to stay on the lawn is one. 
For antipodal lawns, which by definition contain precisely one of 
each pair of opposite points and have length $\pi$, we show this is true except when
the jump length~$\cjump$ is of the form $\pi\frac{p}{q}$ with $p,q$ coprime and $p$ odd. 
For these jump lengths we show the optimal probability is $1 - 1/q$
and construct optimal lawns.   
For a \emph{pair} of antipodal lawns, we show that the optimal probability of 
jumping from one onto the other is $1 - 1/q$ for $p,q$ coprime, $p$
odd and $q$ even, and one in all other cases.  
For an antipodal lawn on the sphere, it is known~\cite{kent2014bloch} that if $\jp = \pi/q$, 
where $q \in \mathbb N$, then the optimal retention probability of $1-1/q$ for 
the grasshopper's jump is provided by a hemispherical lawn. 
We show that in all other cases where $0<\jp < \pi/2$,
hemispherical lawns are not optimal, disproving the
hemispherical colouring maximality hypotheses~\cite{kent2014bloch}.
We discuss the implications for Bell experiments and 
related cryptographic tests.  
\end{abstract}
\maketitle
  
\section{Introduction}

The grasshopper problem was first introduced in 
Ref.~\cite{kent2014bloch}, which analysed Bell
inequalities for the case where two parties carry out spin
measurements about randomly chosen axes and obtain the 
spin correlations for pairs of axes separated by angle $\jp$.     
It was noted \cite{kent2014bloch} that tighter bounds could be
obtained by a version of the following problem on 
the Bloch sphere.   Half the area of a sphere is covered by a lawn, with the property that
exactly one of every pair of antipodal points belongs to the lawn.  
A grasshopper lands at a random point on the lawn, and then jumps 
in a random direction through spherical angle $\jp$.   What lawn shape maximises the 
probability that the grasshopper remains on the lawn after jumping,
and what is this maximum probability (as a function of $\jp$)?  

Ref.~\cite{goulko2017grasshopper} studied the planar version of the grasshopper problem,
giving a combination of analytic and numerical results, and also
discussed several interesting variants.  
As these discussions illustrate, the
grasshopper problem is an appealing problem in geometric combinatorics
which is of intrinsic interest, independent of its original motivations.  

In this paper we consider versions of the
grasshopper problem on the circle and the sphere.
One can similarly motivate this work as an exploration 
of geometric combinatorics in simple manifolds with non-trivial
topologies. Although the circle defines the simplest non-trivial
version of the problem, its solution still has some interesting
features. Results for antipodal lawns on the circle 
also imply some optimality results for antipodal lawns on
the sphere, since a spherical antipodal lawn defines a 
circular one for every great circle.   The spherical version of
the grasshopper problem has some features in common with the planar
version, but the non-trivial topology and compactness mean that 
planar results do not always have direct parallels, and one would 
not expect all regimes of the planar ``phase diagrams'' of Ref.~\cite{goulko2017grasshopper}
to have qualitatively similar parallels in the spherical case.     

Another strong motivation is to develop further the analysis of 
Bell inequalities initiated in Ref.~\cite{kent2014bloch}, which
considered the average anti-correlations attainable by local hidden
variable models for random pairs of measurement axes separated by
a given fixed angle.    
Both the circle and the sphere are relevant here, since the 
circle parametrizes the simplest class of projective polarization measurements
commonly used in Bell experiments on photons, while the sphere parametrizes
all possible projective polarization measurements, or more generally
all possible projective measurements on any physical system defining a
qubit.   We show that, perhaps surprisingly, local hidden variable
models can produce perfect anticorrelations for random pairs of 
axes separated by angle $\cjump$ on the circle, unless $\cjump = \pi
(p/q)$ with $p$ odd, $q$ even and $(p,q)=1$, i.e.\ $p$ and $q$ are coprime.  This means that any imprecision in 
the axis separation, however slight, allows classical simulation 
of quantum correlations.   For the sphere, our results show that
hemispherical lawns are not optimal unless $\jp= \pi (p/q)$, 
where $p$ is odd and $(p,q)=1$.   This means that local hidden
variable models can achieve stronger anticorrelations than 
previously realised for generic $\jp$.   We discuss these
results and their implications further below.

\section{The grasshopper on a circle}

We first consider what seems to be the simplest non-trivial
version of the problem, in which the grasshopper is constrained to
jump through a known fixed angle around a circle of circumference $2 \pi$.   
This also allows us to study the significance of an antipodal 
condition, by considering lawns of length $\pi$ that contain
precisely one of every pair of antipodes. 

Then we consider the case in which there are two potentially
independent (maybe overlapping) antipodal lawns, in which the grasshopper starts
at a random point on one.  As before, it jumps through a 
known fixed angle in a random direction. 
In this case, the question is how to configure the
lawns to maximize the probability that it lands on the second.   

Ref.~\cite{damianthesis} previously discussed the grasshopper
on the circle, with two antipodal lawns, and noted the optimality 
of the lawn $S_{\pi, q}$ for the case of rational jumps with even 
numerator (see below).

These versions of the problem are physically motivated as follows. 
A great circle on the Bloch sphere  defines 
linear polarization measurements, which are easily implemented
and commonly used in Bell and other cryptographic tests.
``Classical'' hidden variable models for these measurement outcomes are defined by 
antipodal colourings of the circle.   A simple hidden variable model
for the quantum singlet state,
in which the outcomes of the same measurements on both subsystem
are perfectly anticorrelated, is defined by a single antipodal colouring.
The most general model, in which measurement outcomes may be
independent, is defined by a pair of antipodal colourings. 
The retention probabilities define the anticorrelations 
predicted by these models, which can be compared to those 
predicted by quantum theory via Bell inequalities.   
As we discuss below, our results have interesting implications for
testing and simulating quantum entanglement.  

\subsection{Statement of the problem}

\qquad{\bf General lawns:} \qquad Following  Ref.~\cite{goulko2017grasshopper}, the most general version
of the grasshopper problem on the circle allows lawns of variable density,
defined by a measurable probability density function 
$f$ on the circle ${\mathbb S}^1$ satisfying $ f (\mathbf{\cphase}) \in
[ 0 , 1 ]$  for all $\mathbf{\cphase} \in \left[ 0 , 2 \pi \right)$ and 

\begin{equation}
\int_{{\mathbb S}^1}d \cphase f  (\mathbf{\cphase}) = L \, . 
\label{eq:munorm}
\end{equation}
Here $L$ is the lawn length.
We take the circle to have length $2 \pi$,
so the non-trivial cases have $0 < L < 2 \pi$.  

It will suffice for most of our discussion to consider indicator functions $f$ with 
$f ( \mathbf{\cphase} ) \in \{ 0, 1 \}$.  
We represent such lawns by measurable subsets $S \subset \left[ 0 , 2
  \pi \right)$, where $S = \{ \cphase \, : \, f ( \mathbf{\cphase} ) =1 \}$
has measure $\mu (S) = L$.   

The functional $\continuousp{f}{\cjump}$ is 
defined by 
\begin{equation}\label{successprob}
  \continuousp{f}{\cjump} = 
\frac{1}{2L} \int_{{\mathbb S}^1}d \cphase [ f(\mathbf{\cphase}) f (\mathbf{\cphase + \cjump} )  
 + f(\mathbf{\cphase}) f (\mathbf{\cphase - \cjump} ) ] \,  . 
\end{equation} 
Here and below all angles are taken modulo $2 \pi$.  
We refer to the expression in Equation~(\ref{successprob}) as the {\it retention probability}: it
defines the probability that a grasshopper starting at a randomly chosen point on the
lawn remains on the lawn after jumping through angle $\phi$ in a random direction. 

The grasshopper problem is then to answer the following.   
What is the supremum of $\continuousp{f}{\cjump}$ over all such
functions $f$, for each value of $\cjump \in \left[0, 2 \pi
\right)$?  Which $f$, if any, attain the supremum?  Or if none, which
sequences approach the supremum value?

We will show below that the supremum value is $1$ for 
all $\cjump \in \left[0, 2 \pi
\right)$.    

\qquad{\bf Antipodal lawns:}  \qquad For the case $L = \pi$, we also consider these questions restricted to 
antipodal lawns, for which $f ( \cphase ) \in \{ 0 , 1 \}$ and  $f (
\cphase + \pi ) = \overline{f( \cphase )} = 1 - f (\cphase )$
for all $\cphase$. 

We will show below that the supremum value is $1$ for 
all $\cjump \in \left[0, 2 \pi
\right)$ except $\cjump$ of the form $\pi \frac{p}{q}$ where $p$ and $q$
are coprime and $p$ is odd, when it is $1 - \frac{1}{q}$.

\qquad{\bf Two independent antipodal lawns:} \qquad  We can extend
the grasshopper problem to the case of two lawns given by 
antipodal measurable subsets $S^A , S^B$ of the circle, defined by
suitable indicator functions $f^A , f^B \, : \, \left[ 0, 2 \pi
\right) \rightarrow \{ 0, 1 \}$, where $S^X = \{ \cphase :
f^X ( \cphase ) = 1 \}$.   
Here the relevant functional for jump $\cjump$ is 
defined by 
\begin{equation}\label{successprobtwolawns}
  \continuousp{f^A , f^B }{\cjump} = 
\frac{1}{2L} \int_{{\mathbb S}^1}d \cphase [ f^A (\mathbf{\cphase}) f^B (\mathbf{\cphase + \cjump} )  
     + f^A (\mathbf{\cphase}) f^B (\mathbf{\cphase - \cjump} ) ]  \,  . \nonumber
\end{equation} 
As we explain below, this is the version of the problem relevant to analysing general local
hidden variable models for bipartite quantum states, where the
measurements chosen on the two subsystems are independent and 
each parametrised by the circle.  
We will show below that the supremum value is $1$ for 
all $\cjump \in \left[0, 2 \pi
\right)$ except $\cjump$ of the form $\pi \frac{p}{q}$ where $p$ and $q$
are coprime with $p$ odd and $q$ even, when it is $1 - \frac{1}{q}$.    

This illustrates that the two-lawn version
of the problem is a non-trivial extension of the one-lawn version, 
even in the case of the circle.   In particular,  
for jump values of the form $\pi \frac{p}{q}$ with $p$ and $q$ coprime
and both odd, the optimal one-lawn jump probability is 
strictly less than $1$, while the optimal two-lawn jump probability is $1$.   

\qquad{\bf One-directional grasshopper:} \qquad A natural and apparently simpler version of the grasshopper problem in
one dimension is given by assuming that the grasshopper always 
jumps in the same direction.   For a grasshopper on the 
real line \cite{goulko2017grasshopper}, this restriction makes no essential difference:
the supremum probability is still $1$ for any jump length, and the
same limiting construction works for the one-directional and 
standard bidirectional case.   

Similarly, one can assume that the grasshopper on the circle only
jumps clockwise (or anticlockwise).   For one lawn, our results and constructions
for the bidirectional case carry over straightforwardly to this
case, so we will not discuss it separately. 
With two lawns, the construction 
$S^A = \left[ 0, L \right)$, $S^B = \left[ \cjump, L + \cjump \right)$ 
obviously gives retention probability $1$ for clockwise jump $\cjump$. 

For $L = \pi$ this construction defines antipodal lawns.  
The one-directional grasshopper problem on the circle thus turns out
to be of no independent interest, and we will not discuss it further.  

\subsection{Solution for general lawns}

\begin{lemma}[rational case] The optimal retention probability is $1$ for
a lawn of length $L$ with jump $\cjump=\frac{p}{q} 2 \pi$, 
where the highest common factor $(p,q)=1$.    
\end{lemma}

\begin{proof} 
Define the lawn $S_{L,q}$ to be 
\begin{equation}\label{lawnln}
S_{L,q} = \bigcup_{j=0}^{q-1} \left[  2 \pi  \frac{j}{q} ,  2 \pi  
  \frac{j}{q} + \frac{L}{q}  \right) 
\end{equation}
An example of such a lawn in depicted in
Figure~\ref{fig:general_rational}. Clearly $S_{L,q}$ has length $L$ and
retention probability~$1$.  
\end{proof}

This solves the problem for rational jumps.
\begin{figure}
\includegraphics[width=0.4\textwidth]{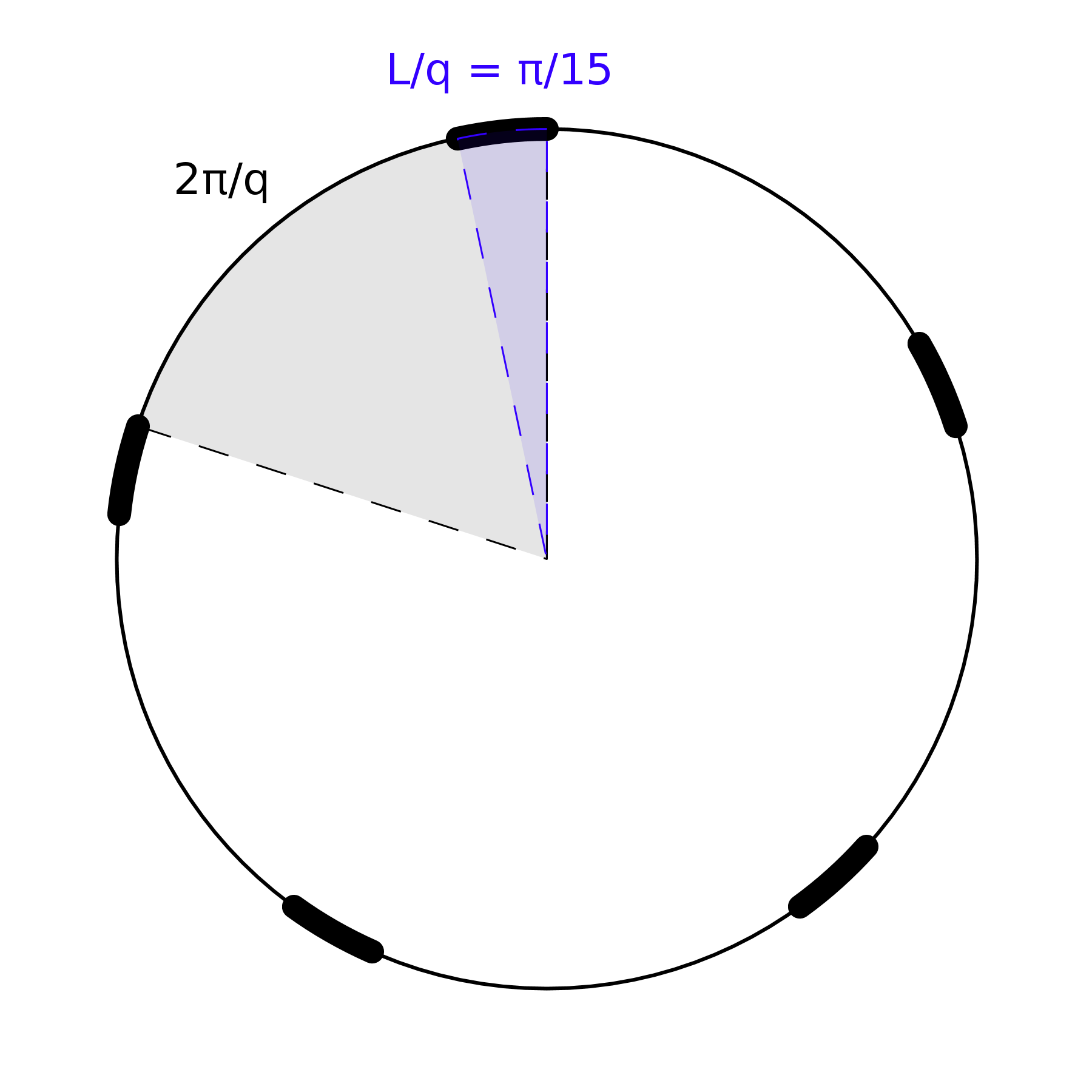}
\caption{\label{fig:general_rational}Sketch of an optimal general lawn
  $S_{L,q}$ with $L=\pi/3$,
for a rational jump $\cjump=\frac{p}{q}2 \pi$ with $p=2$ and $q=5$.}
\end{figure}

\begin{lemma}[irrational case]
The supremum retention probability is $1$ for a lawn of length $L$ 
when the jump $\cjump$ is an irrational multiple of $2 \pi$.    
\end{lemma}
\begin{proof} 
Define the lawn $S_{\epsilon, \cjump}^{k}$ to be 
\begin{equation}
S_{\epsilon,\cjump}^{k} = \bigcup_{j=0}^{k-1} \left[ j \cjump , j \cjump + \epsilon
\right) \, , 
\end{equation}
where (following the usual definition of set union) 
any overlapping intervals are counted only once.
Let $K$ be the largest value of $k$ such that 
$S_{\epsilon,\cjump}^{k}$ has length smaller than $L$. 
Weyl's equidistribution theorem \cite{weyl1910gibbs} guarantees that such a $K$
exists for any $\epsilon$ in the range $L > \epsilon > 0$. 
Now define the lawn $S_{L, \epsilon, \cjump}$ by 
\begin{equation}
S_{L, \epsilon, \cjump} = S_{\epsilon,\cjump}^{K} \cup \left[K \cjump, K \cjump +
  \delta \right) \, ,
\end{equation}
where $\delta \in \left(0,\epsilon\right]$ is the smallest
value in that range such that $S_{L,\epsilon, \cjump}$ has length $L$.  
Unless the grasshopper starts in the first interval and jumps
anticlockwise, or the last interval and jumps clockwise, it
remains on the lawn, which thus has retention probability $\geq 1 - \epsilon$.
A sequence with $\epsilon \rightarrow 0$ gives us the
supremum value of $1$.
\end{proof} 

Thus we obtain
\begin{theorem}
For unconstrained lawns of length $L$ the supremum retention probability is 1. The supremum is attained for rational jumps $\cjump=\frac{p}{q} 2 \pi$.\end{theorem}

\subsection{Solution for antipodal lawns}

\begin{lemma}[rational case, even numerator]
The optimal retention probability is $1$ for an antipodal lawn 
with jump $ \cjump = \pi \frac{p}{q} $,
where $(p,q)=1$ and $p$ is even. 
\end{lemma}
\begin{proof}
An antipodal lawn has length $\pi$. 
Since $(p,q)=1$ and $p$ is even, $q$ is odd. 
From Eqn. (\ref{lawnln}), we have 
\begin{equation}
S_{\pi,q} = \bigcup_{j=0}^{q-1} \left[ 2 j  \pi \frac{1}{q} , (2j+1) \pi
  \frac{1}{q} \right) 
\end{equation}
Clearly $S_{\pi,q}$ has length $\pi$, is antipodal, and has retention
probability $1$.  
\end{proof} 
An example for such a lawn in depicted in Figure~\ref{fig:antipodal_rational_even}.
\begin{figure}
\includegraphics[width=0.4\textwidth]{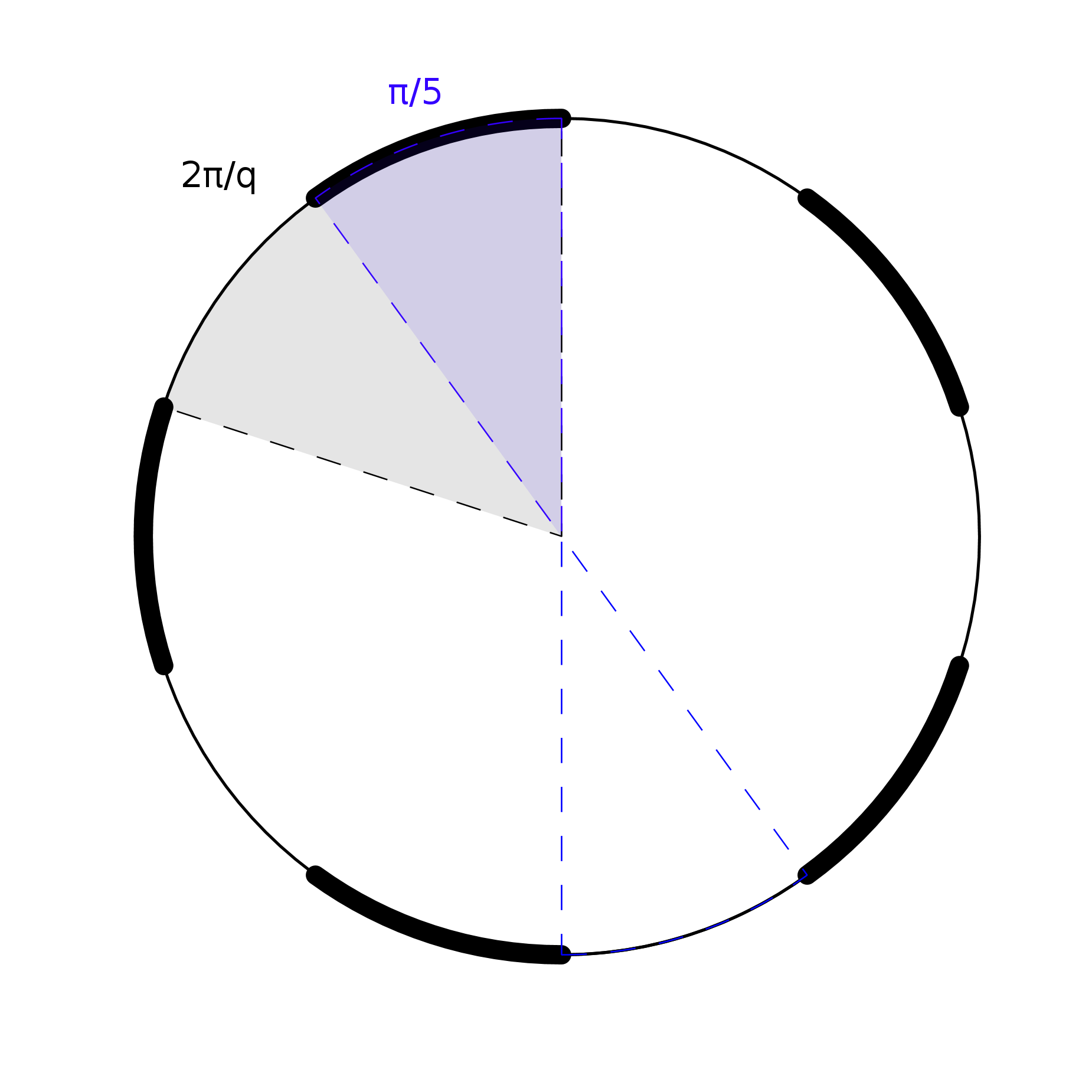}
\caption{\label{fig:antipodal_rational_even}Sketch of an optimal
  antipodal lawn $S_{\pi,q}$ for a rational jump $\cjump=\pi\frac{p}{q}$
  with even numerator $p=2$ and $q=5$.}
\end{figure}

Now consider a lawn of length $\pi$ with jump $ \cjump = \pi \frac{p}{q} $,
where $(p,q)=1$ and $p$ is odd.   We take $q \geq 2$ here; we consider
the special case of $q=1$ (which for $p$ odd implies a jump through
angle $\pi$) separately below.  

\begin{lemma}[rational case, odd numerator]
The optimal retention probability is $(1- \frac{1}{q})$
for an antipodal lawn with  jump $ \cjump = \pi \frac{p}{q} $,
where $(p,q)=1$, $p$ is odd and $q \geq 2$.
\end{lemma}

\begin{proof} 
First, consider lawns defined by indicator functions $f$ with 
$f ( \mathbf{\cphase} ) \in \{ 0, 1 \}$.
Consider any starting point $\cphase$ and the set of 
points $\{ \cphase_k  = \cphase + k (  \pi \frac{p}{q} ) : 0 \leq k \leq
(2q -1) \}$.
These points are all distinct, and the points 
$\cphase_k , \cphase_{k + q}$ are antipodal, where the sum is
taken modulo $2q$.  Hence any antipodal lawn will contain precisely
$q$ of these $2q$ points. 
A sequence of $q$ jumps in either direction takes a point on
the lawn to the antipodal point, which must be off the lawn.   
Thus, whichever $q$ points are on the lawn, a clockwise jump
from at least one of them takes the grasshopper off the lawn, 
and similarly for anti-clockwise jumps. 
Because this is true for all such sets of lawn points, the probability
of leaving the lawn is at least $\frac{1}{q}$, i.e. the retention
probability is at most $1 - \frac{1}{q}$. 

For a general indicator function, consider again  
the discrete set of 
points $D_{\cphase,\frac{p}{q}}=\{\cphase_k=\cphase + k(\pi\frac{p}{q}):0\leq k\leq(2q -1)\}$, and suppose
$f(\cphase_k)=p_k$, with $p_k+p_{k+q}=1$ and 
$0\leq p_k\leq 1$.   
Suppose the grasshopper first lands at a point in $D_{\cphase,\frac{p}{q}}$.   It then jumps to another point in 
$D_{\cphase,\frac{p}{q}}$.
In what follows we take all probabilities conditioned on the first landing being in  
$D_{\cphase,\frac{p}{q}}$. 

Let $p_k$ be the probability of first landing at $\cphase_k$,
so we have $p_k + p_{k+q}=1$ since the lawn is antipodal.   
The retention probability is 
\begin{eqnarray*}
\frac{1}{2q}\sum_{k=0}^{2q-1} p_k(p_{k-1}+p_{k+1})&=&\frac{1}{q}\sum_{k=0}^{2q-1}p_k p_{k+1}=\frac{1}{q}\sum_{k=0}^{q-1}\left( p_k p_{k+1}+p_{k+q}p_{k+q+1}\right)\\
&=&\frac{1}{q}\left(p_{q-1}(1-p_0)+(1-p_{q-1})p_0+\sum_{k=0}^{q-2}\left( p_k p_{k+1}+(1-p_k)(1-p_{k+1})\right)\right) \,  .
\end{eqnarray*} 
The retention probability as a function of $p_0,\ldots ,p_{q-1}$ is linear in 
each variable, therefore its maximum value is attained for some choice 
of the extreme values $0$ or $1$ for each variable. This consideration  
returns us to the case with indicator functions previously discussed.

In the case $p=1$, the semi-circular lawn 
\begin{equation}
S_{{\rm semi}} = \left[ 0 , \pi \right ) , 
\end{equation}
which has retention probability $1 - \frac{1}{q}$, is thus optimal.    

For general odd $p$, the lawn
\begin{equation}
S_{\pi,p,q} = \bigcup_{j=0}^{q-1} \left[  j  \pi \frac{p}{q} , j \pi
  \frac{p}{q} + \frac{\pi}{q} \right) 
\end{equation}
is antipodal and has retention probability $1 - \frac{1}{q}$, 
and is thus optimal.    
\end{proof}

There are also other optimal lawns that are not rotated versions of
these.   
For general $p$, we can construct a lawn with the maximal retention probability as follows. 
\begin{equation}
S'_q = \bigcup_{j=0}^{q -1 } \left[  \pi \frac{jp}{q} , \pi
  \frac{jp}{q} + \frac{\pi}{2q} \right) \cup 
\bigcup_{j=0}^{q -1 } \left[  -  \pi
  \frac{jp}{q} - \frac{\pi}{2q} , -  \pi \frac{jp}{q} \right) \, .
\end{equation}
Here the two unions define disjoint ``demi-lawns'' of length
$\frac{\pi}{2}$. Several examples are plotted in Figure~\ref{fig:antipodal_rational_odd}.
\begin{figure}
\includegraphics[clip, trim=3cm 4mm 4cm 3mm, width=0.33\textwidth]{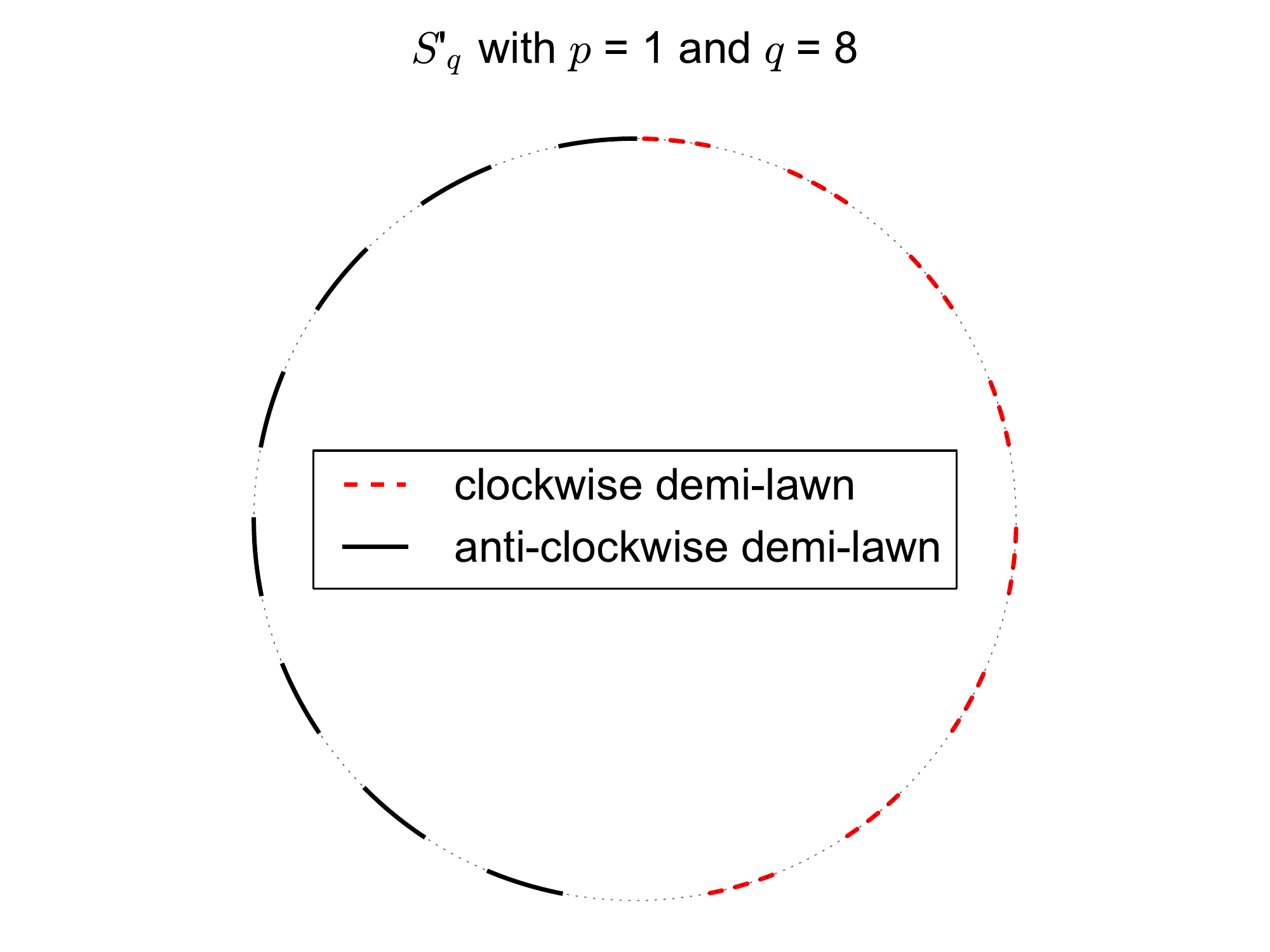}\hfill
\includegraphics[clip, trim=3cm 4mm 4cm 3mm, width=0.33\textwidth]{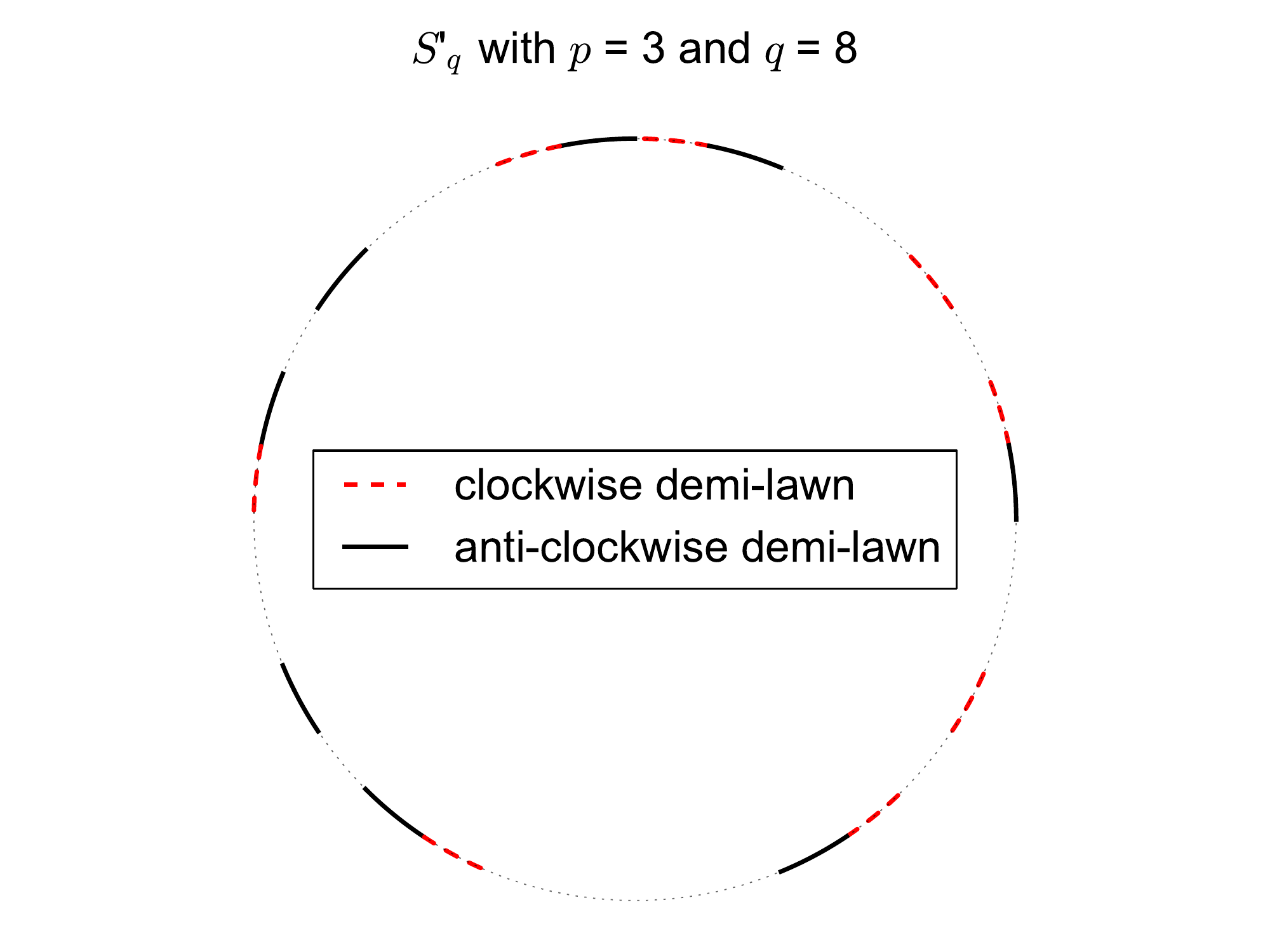}\hfill
\includegraphics[clip, trim=3cm 4mm 4cm 3mm, width=0.33\textwidth]{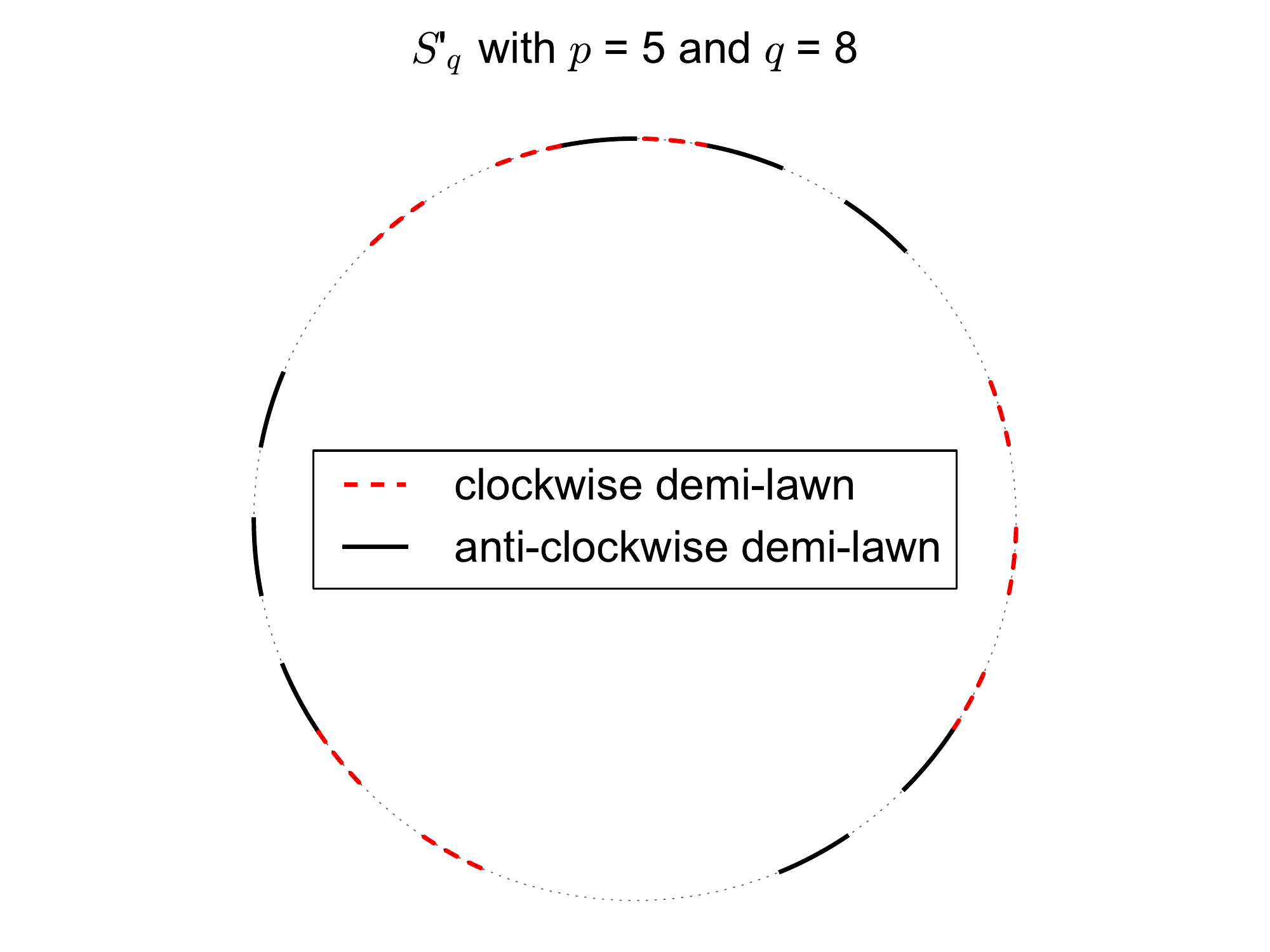}
\caption{\label{fig:antipodal_rational_odd}Sketch of an optimal
  antipodal lawn $S'_{q}$ for a rational jump $\cjump=\pi\frac{p}{q}$
  with odd numerator $p=1$ (left), $p=3$ (middle), $p=5$ (right), and
  (in all panels) $q=8$. The two demi-lawns are marked by different
  line styles (black solid lines and red dashed lines).}
\end{figure}

The special case of a jump through angle $\pi$ takes any point to
its antipodal point. In the case when the indicator function is restricted to 
$f ( \mathbf{\cphase} ) \in \{ 0, 1 \}$, all antipodal lawns have probability 0. 
However, a higher probability can be attained for general indicator functions.
\begin{lemma}[jump through angle $\pi$]
The optimal retention probability is $\frac{1}{2}$ for an antipodal lawn with jump $\cjump=\pi$.\footnote{We thank Carlo 
Piosevan for noting this special case.} 
\end{lemma}
\begin{proof}
   Since antipodal points have lawn density
$p$ and $(1-p)$ for some $p \in \left[ 0 , 1 \right]$, the optimal
density for each pair is given by $p=\frac{1}{2}$, which
maximizes $p(1-p)$.   The optimal lawn in this case thus
has uniform density $\frac{1}{2}$ and the retention probability is $\frac{1}{2}$. 
\end{proof}

\begin{lemma}[irrational case]
The supremum retention probability is $1$ for antipodal lawns 
with jump angle $\cjump =  \pi x$, when $x$ is
irrational.
\end{lemma}

\begin{proof}
By Hurwitz's theorem \cite{hurwitz1891angenaherte}
there are infinitely
many rationals $\frac{p}{q}$ such that
\begin{equation}
\left| x - \frac{p}{q} \right| \leq \frac{1}{\sqrt{5} q^2 } \, .
\end{equation}
We could apply this result directly, considering separately
the possibilities that $p$ is odd or even. 
A slightly shorter argument follows from an extension of
Hurwitz's theorem due to Uchiyama \cite{uchiyama1980rational}, 
elaborated by Elsner \cite{elsner1996approximation}.  This implies that
for any irrational $x$ there are infinitely many rationals 
$\frac{p}{q}$ with $p$ even and $q$ odd such that
\begin{equation}
\left| x - \frac{p}{q} \right| \leq \frac{1}{q^2 } \, .
\end{equation}
Let $\frac{p}{q}$ be a rational approximation to $x$ of
this type.    For jump $x$, the lawn $S_{\pi,q}$ has leaving probability
\begin{equation}
q \left| x - \frac{p}{q} \right|   \leq   \frac{1}{q} \, .
\end{equation}
By considering an infinite sequence of approximations with increasing $q$, we can
make this arbitrarily small.    Hence the 
supremum retention probability is $1$. 
\end{proof}

Combining the lemmas for antipodal lawns we obtain
\begin{theorem}
For antipodal lawns with general indicator functions, the supremum retention probability is 1, except for jumps of the form $\cjump=\pi\frac{p}{q}$, where $(p,q)=1$ and $p$ is odd. In the latter case, it is $1/2$ for $q=1$ and $1-\frac{1}{q}$ otherwise.
\end{theorem}

\subsection{Two antipodal lawns}

We now consider the case of two lawns, $S^A$ and $S^B$, 
where the grasshopper starts on one and we are interested
in optimizing the probability that it lands on the other. 
In versions of the problem where the supremum probability
is $1$ for a single lawn, we can obviously obtain this 
supremum by taking $S^A = S^B$.    
The only case that remains of interest is thus two
antipodal (thus length $\pi$) lawns with jump $ \cjump = \pi \frac{p}{q} $,
where $(p,q)=1$ and $p$ is odd.   

\begin{lemma}  
The optimal retention probability for two antipodal lawns
is $(1 - \frac{1}{q})$ when the jump $ \cjump = \pi \frac{p}{q} $,
where $(p,q)=1$, $p$ is odd and $q$ is even.
\end{lemma} 
\begin{proof}
We have that $2 \cjump = \pi \frac{p}{q'}$, $2q' = q$ and $(p,q')=1$.
Again, we first consider the case in which
the lawn densities take only values $0$ or $1$.
Any such lawn is defined by the set $S= \{ \cphase \, : \, f (\cphase ) = 1
\, \}$ .  Define the complementary lawn 
$$ 
\overline{S} = \left[ 0 , 2 \pi \right ) \setminus S \, .
$$
If $S$ is antipodal, then so is $\overline{S}$, since 
$\overline{S}$ is the set of points antipodal to points in $S$. 

Consider any starting point $\cphase$ and the sets of 
points 
$$A_{\cphase} = \{ \cphase^A_k  = \cphase + k (  \pi \frac{p}{q'} ) : 0 \leq k \leq
(2q' -1) \} \, , $$
$$B_{\cphase} = \{ \cphase^B_k  = \cphase + (k+ \frac{1}{2} ) (  \pi \frac{p}{q'} ) : 0 \leq k \leq (2q' -1) \} \, . $$ 
These points are all distinct.   The points 
$\cphase^A_k , \cphase^A_{k + q'}$ are antipodal, as are the points
$\cphase^B_k , \cphase^B_{k + q'}$, where the subscript sums are 
taken modulo $2q'$.  Hence the antipodal lawns $S^A , S^B$ contain precisely
$q'$ of the $2q'$ points $A_{\cphase} , B_{\cphase}$ respectively. 
Starting at a point on $S^A$, a sequence of $2q'$ jumps clockwise through $\cjump$ 
reaches the antipodal point, which must be off the lawn.   
Thus at least one of these jumps links a point on $S^A$ to a point off
$S^B$ or a point on $S^B$ to a point off $S^A$.    
In the first case, there is at least one clockwise jump from a point
in $A_{\cphase} \cap S^A$ to 
$B_{\cphase} \cap \overline{S^B}$.   In the second case, there is at least
one clockwise jump from $B_{\cphase} \cap S^B$ to $A_{\cphase} \cap
\overline{S^A}$. 
As the lawns are antipodal, this implies at least one anticlockwise
jump from $A_{\cphase} \cap S^A$ to $B_{\cphase} \cap \overline{S^B}$. 
Thus, whichever $q'$ points lie in 
$A_{\cphase} \cap S^A$, there is at least one jump from one of them
in one direction that takes the grasshopper to $\overline{S^B}$, 
Because this is true for all such sets of lawn points, the probability
of leaving the lawn is at least $\frac{1}{2q'}$, i.e. the retention
probability is at most $1 - \frac{1}{2q'} = 1 - \frac{1}{q}$. 

For the general case with unrestricted density functions, we
can argue as before.  If lawn $A$ has a point $\cphase$ 
such that $f_A (\cphase) = 1 - f_A (\cphase+ \pi) =p$, where $0 < p <1$,
then at least one of the choices  $p=0$ and $p=1$ 
does not decrease the retention probability.  
The same argument holds for lawn $B$.  We can thus restrict without
loss of generality to lawns with density $0$ or $1$ everywhere.

In the case $p=1$, where $\cjump = \pi \frac{1}{2q'} = \pi \frac{1}{q}$,
we can attain this retention probability by 
taking $S^A = S^B = \left[ 0 , \pi \right) $ (i.e. identical semicircular
lawns). A jump from $S^A$ clockwise remains on $S^A$
unless it begins in the segment $\left[ \pi ( 1 - \frac{1}{q} ) , \pi \right)$;
similarly an anticlockwise jump remains on $S^A$ unless it 
begins in the segment $\left[ 0 , \pi \frac{1}{q} \right)$.  

For general $p$ odd, we can construct a lawn with the same 
retention probability by taking $S_A = S_B = S'_q$, where
$S'_q$ is the antipodal lawn defined above, which has
leaving probability $\frac{1}{q}$.   
\end{proof}

\begin{lemma}
The optimal retention probability for two antipodal lawns
is $1$ when the jump $ \cjump = \pi \frac{p}{q} $,
where $(p,q)=1$, $p$ is odd and $q$ is odd. 
\end{lemma}
\begin{proof}
We have $ \cjump = \pi \frac{p}{q}$.
Define the lawns 
$$ S_A = S_{\pi, q} = \bigcup_{j=0}^{q-1} \left[ \pi \frac{2j}{q} , \pi
  \frac{2j+1}{q} \right ) \,  $$ 
and 
$$ S_B = \overline{S_A} = \bigcup_{j=0}^{q-1} \left[ \pi \frac{2j+1}{q} , \pi
  \frac{2j+2}{q} \right ) \, . $$ 
These lawns are antipodal, and the jump probability from $S_A$ 
to $S_B$ for a jump of angle $\cjump$ is $1$, so this
lawn configuration is optimal.  
\end{proof}
 
Combining the lemmas for a pair of antipodal lawns we obtain
\begin{theorem}
For antipodal lawns the supremum retention probability is 1, except for jumps of the form $\cjump=\pi\frac{p}{q}$, where $(p,q)=1$, $p$ is odd and $q$ is even. In the latter case, it is $1-\frac{1}{q}$.
\end{theorem}


\section{The grasshopper on a sphere}

\subsection{Statement of the problem}

We now consider the spherical (one-lawn, antipodal) version of the problem.
The lawn is now a subset $L$ of the sphere, $\mathbb S^2$, in three-dimensional space,
that is antipodal: every point $x$ of the sphere belongs to $L$
if and only if the opposite point does not belong to $L$.
As before, the grasshopper starts at a point chosen
uniformly at random in the lawn, and jumps a fixed distance $\jp$
in a direction chosen uniformly at random. The goal is to pick
the shape of $L$ that maximizes the probability of
a successful jump, i.e., the probability of staying
on the lawn. Put differently, the goal is to maximize the integral
\begin{equation}\label{spheresuccess}
    \int_{\mathbb S^2} f(x)\, ds \int_{C_{\jp}(x)} f(y)\, d\omega,
\end{equation}
where the point $x$ corresponds to the surface element $ds$
on the sphere $\mathbb S^2$,
the point $y$ is taken from the circle~$C_{\jp}(x)$ of radius~$\jp$ around the point $x$,
the angle $\omega \in [0, 2\pi)$ is the corresponding position on the circle,
and the function
$f \colon \mathbb S^2 \to \{0, 1\}$ is defined by
$f(x) = 1$ if and only if $x \in L$.

Theorems~\ref{thm:opt} and~\ref{thm:main} below show that
the retention probability, which is proportional to (\ref{spheresuccess}), is maximised by the hemispherical lawn
if and only if $\jp = \pi/q$, with $q>1$ integer.

\subsection{Corollaries of results on the circle}

\begin{theorem}\label{thm:opt}
The optimal retention probability is 
$(1 - \frac{1}{q}) $
for an antipodal lawn on the sphere with jump  $\frac{\pi}{q}$,
where $q\geq 2 $ is a positive integer.
\end{theorem}

\begin{proof} 
For a single antipodal lawn on the circle, with jump $\frac{\pi}{q}$
for positive integer $q$,
we have shown that the optimal
retention probability is $1 - \frac{1}{q}$ and is
attained by a semi-circular lawn.  

Now consider an antipodal lawn on the sphere with the same jump
$\frac{\pi}{q}$. 
Any starting point on the lawn and any jump direction
together define a great circle.   The start point
is equally likely to be in any arc of given length
on the circle.   Hence our argument for the upper bound
$1 - \frac{1}{q}$ for antipodal lawns on the circle 
also applies to antipodal lawns on the sphere. 
\end{proof}

Since a hemispherical lawn attains this bound,     
we see that hemispherical lawns are optimal for jumps $\frac{\pi}{q}$,
where $q \geq 2$ is a positive integer.    
We show below that hemispherical
lawns are {\it not} optimal for any other jump values. 

For a pair of antipodal lawns on the circle, with jump $\frac{\pi}{q}$
for even positive integer $q$, we have shown that the optimal
retention probability is 
$1 - \frac{1}{q}$ and is
attained by a pair of identical semi-circular lawns.  
Again, our arguments extend to a pair of antipodal lawns on the
sphere: the optimal retention probability is 
$1 - \frac{1}{q}$ and is
attained by a pair of identical hemi-spherical lawns.  
This result was proved in Theorem 1 of Ref. \cite{kent2014bloch}.

\subsection{Construction of lawns with greater retention probability
  than the hemisphere for $\jp \neq \frac{\pi}{q}$}

We now consider jumps $\jp \neq \frac{\pi}{q}$ and show that
hemispherical lawns are not optimal in these cases, by constructing
lawns that have higher retention probability.
Each of the lawns we construct will be a hemisphere with a ``cogged'' boundary 
akin to the construction in~\cite{goulko2017grasshopper}. The number, size and layout of the cogs 
varies according to the size~$\jp$ of the grasshopper's jump. We consider three 
cases according as~$\jp/ \pi$ is irrational, rational with even numerator, or rational 
with odd numerator. These cases are resolved in Lemmas~\ref{l:irrat}, \ref{l:peven} and~\ref{l:podd}, respectively. 

Lemma~\ref{l:peven} will deal with the simplest case: $\jp=p\pi/q$ where $p$ and $q$ are coprime 
and $p$ is even. This construction is illustrated in Figure~\ref{fig:cc} (left) for $\jp= 6 \pi / 13$.
Starting from any point on the equator, take the set of $q$ points 
visited by a sequence of $q$ consecutive jumps of size~$\jp$ travelling round 
and round the equator $p/2$ times before returning to the starting point. 
Take a second sequence with an antipodal starting point and continuing with the sequence 
of points antipodal to the first sequence. Note that these $2q$ points are distinct and 
separated from each other by a distance of at least $\pi/q$. 
The points are spaced evenly, distance $\pi/q$ apart, round the equator.
Draw circles of sufficiently small radius $r = r(\jp)$ around these $2q$ points such that 
there are no overlaps between them. 
Our lawn consists of the southern hemisphere with these circles added 
for the first sequence and removed for the second sequence.  In other words, we modify 
the hemisphere to fill semi-circular \emph{caps} above the equator and remove semi-circular 
\emph{cups} below the equator, creating our cogged hemisphere.

\begin{figure}[thp]
\begin{center}
\includegraphics[width=0.45\textwidth]{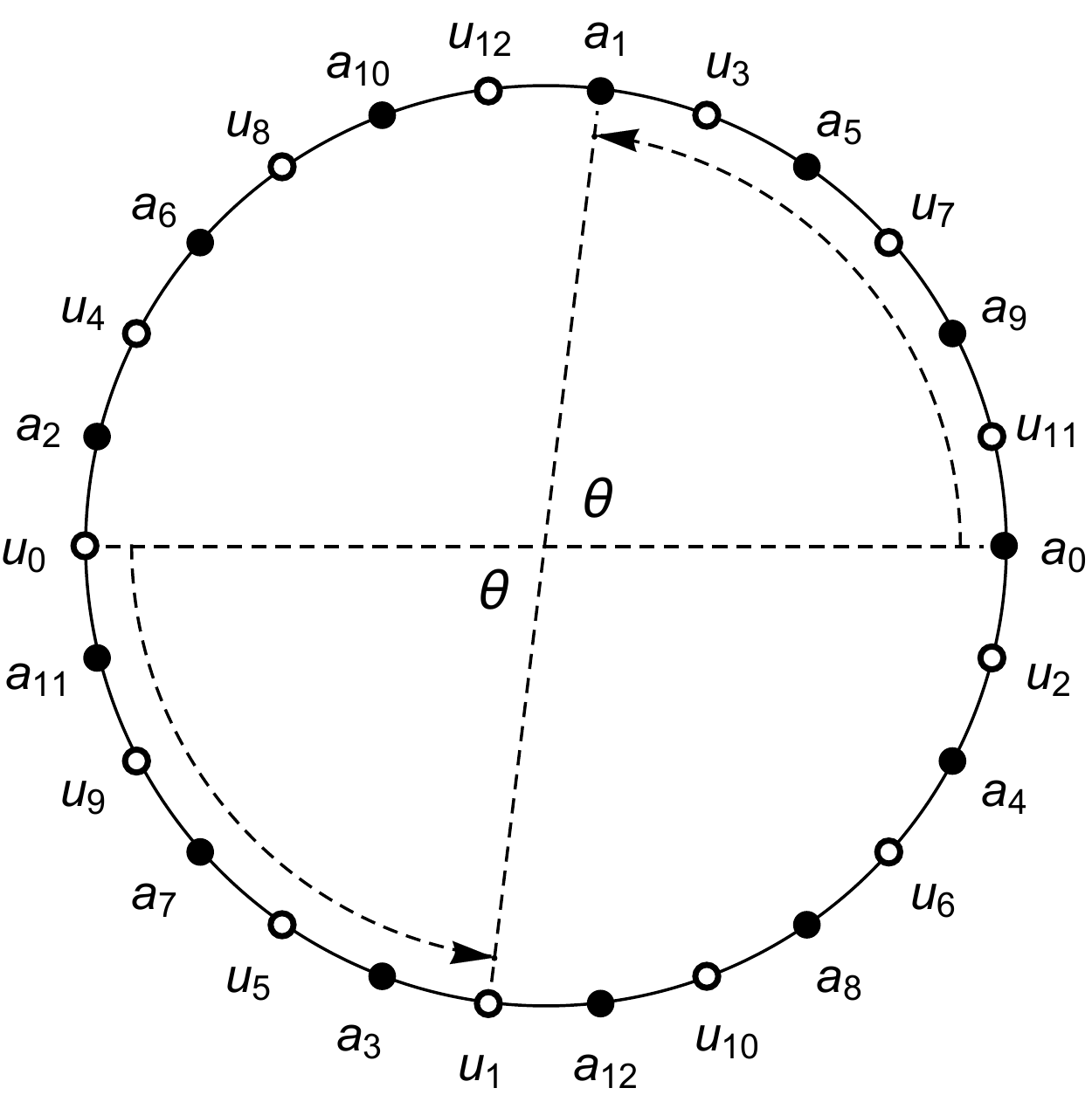}
\hspace*{0.08\textwidth}
\includegraphics[width=0.45\textwidth]{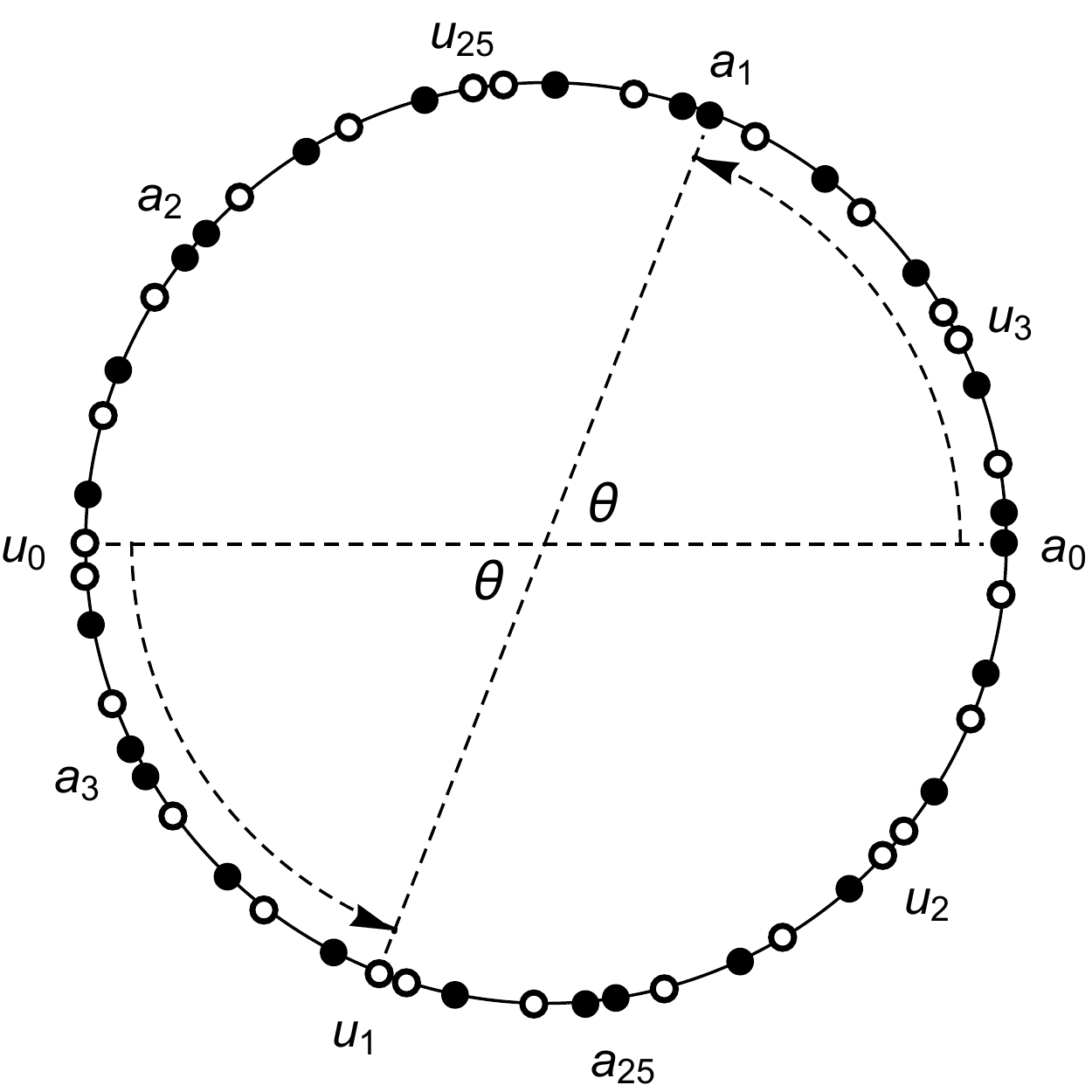}
\end{center}
\caption{Construction of modified lawns for rational $\jp=p\pi/q$ and $p$ even, with $p=6$ and $q=13$ (left panel, Lemma \ref{l:peven}) and for irrational $\jp/\pi$ with $\jp=1.2$ (right panel, Lemma \ref{l:irrat}) The caps and cups around the equator are shown as black and white dots respectively.}
\label{fig:cc}
\end{figure}

For comparison, the right panel of Figure~\ref{fig:cc} looks ahead to the corresponding construction when $\jp/\pi$ is irrational. 
This case will be considered in Lemma~\ref{l:irrat}. 

Our construction ensures that a jump of length~$\jp$ cannot connect any point 
in a cap to any point in a cup, or vice versa. For Lemma~\ref{l:peven}, most of our analysis involves 
estimating the probability of a jump between caps or between cups. For Lemmas~\ref{l:irrat} and~\ref{l:podd} 
we need also to consider jumps from caps or cups to the whole hemisphere $S$;
we describe the modified constructions in the proofs of these lemmas.
Lemma~\ref{l:podd} requires the further complication of having caps and cups of varying sizes.

\subsection{Preliminaries}

\begin{figure}[t]
\begin{center}
\includegraphics[width=0.4\textwidth]{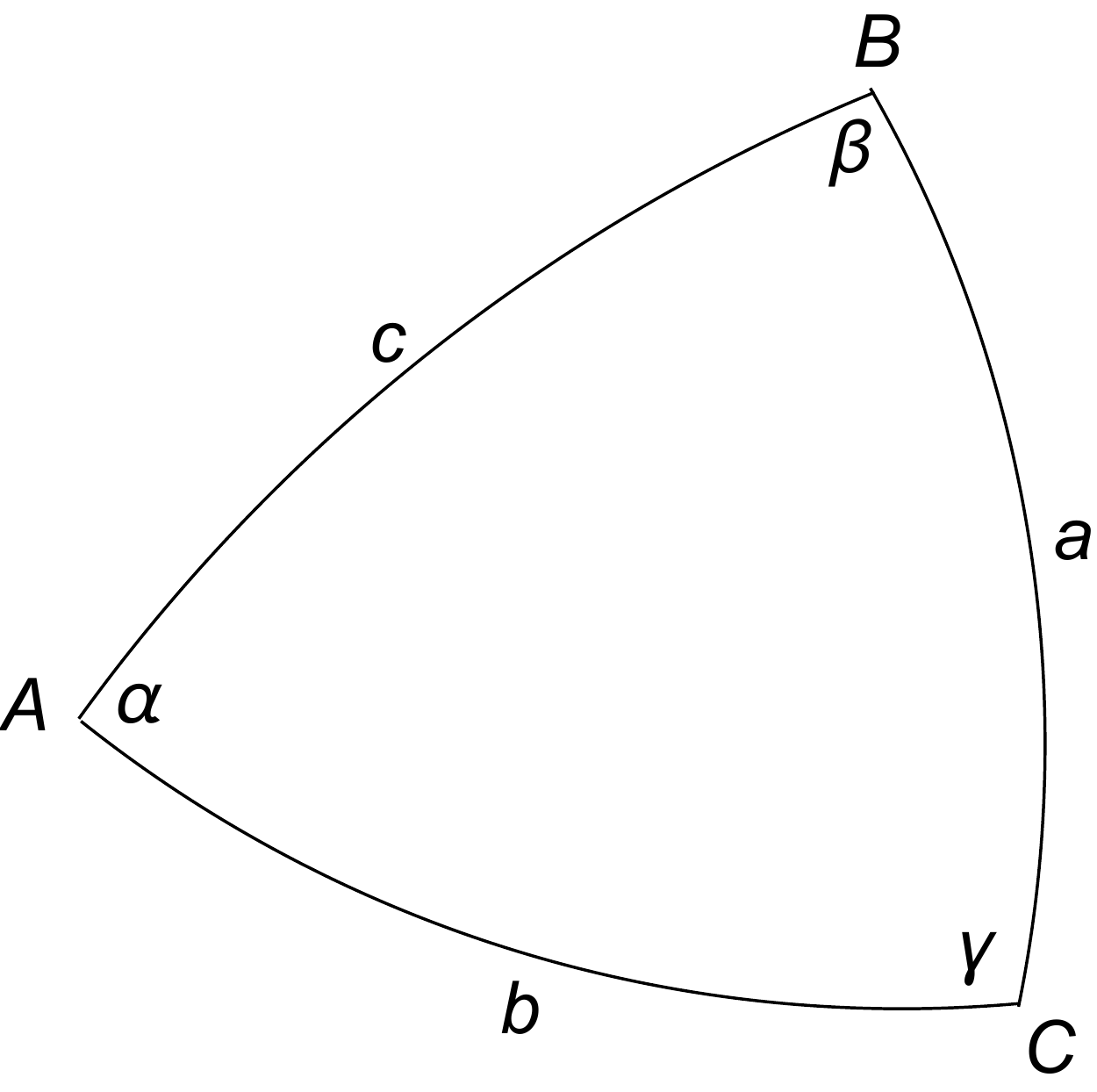}\hfill
\includegraphics[width=0.4\textwidth]{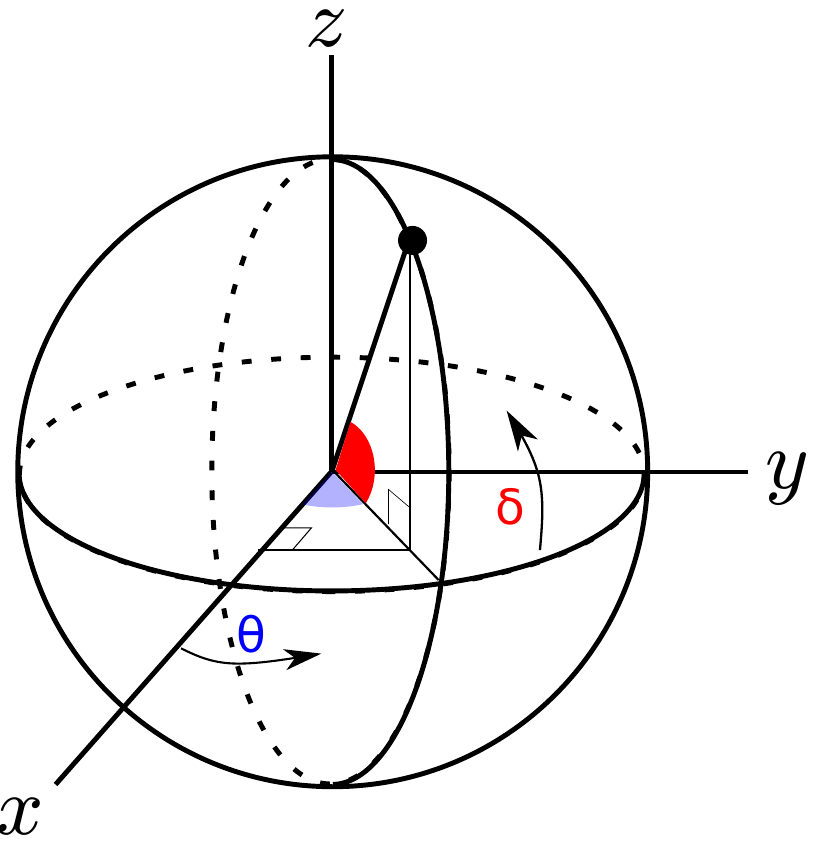}
\end{center}
\caption{Spherical triangle (left) and spherical coordinates (right).}
\label{fig:triangle}
\end{figure}

We recall some basic formulae in spherical geometry. 
Using the notation in Figure~\ref{fig:triangle}, 
where $a,b,c$ are (great-circle) lengths and $\alpha,\beta,\gamma$ are angles, 
we have the \emph{sine rule}:
\begin{equation}\label{eq:sin}
\frac{\sin a}{\sin\alpha} = \frac{\sin b}{\sin\beta} = \frac{\sin c}{\sin\gamma} , 
\end{equation}
the \emph{cosine rule}:
\begin{equation}\label{eq:cos}
\cos a = \cos b \cos c + \sin b \sin c \cdot \cos\alpha ,
\end{equation}
and, in particular, 
\begin{equation}\label{eq:rtcos}
\mathrm{if\ } \alpha = \pi/2 \mathrm{\ then\ }\cos a = \cos b \cos c .
\end{equation}
The next formula follows easily:
\begin{equation}\label{eq:rttan}
\mathrm{if\ } \alpha = \pi/2 \mathrm{\ then\ }\tan a \cos \beta = \tan c .
\end{equation}

For points on the sphere, it is convenient to use pairs of angles $(\lo,\la)$ 
based on (longitude, latitude), where $\lo$ is the azimuthal angle and $\la$ 
is the \emph{elevation} or \emph{copolar} angle, see Figure~\ref{fig:triangle}. 
We take the radius of the sphere to be unity, so a typical point on the surface with spherical coordinates $(\lo,\la)$
has cartesian coordinates $(x,y,z) = (\sin\lo\cos\la, \cos\lo\cos\la, \sin\la)$, 
where the $z$-axis passes through the poles. 

For $B$ and $C$ points on the sphere, the dot-product of their cartesian coordinates gives their angular distance. In Figure~\ref{fig:triangle},
\begin{equation}\label{eq:dist}
B\cdot C = \cos a .
\end{equation}

We recall the notation for the cosecant function: $\csc a = 1 / \sin a$.

In Section~\ref{sec:podd} we make use of the following trigonometric identities and corollaries. 
\begin{proposition}\label{prop:sumcos}
For all integers $q \ge 2$, $$\sum^{q-1}_{j=0} \cos\left(\frac{2j \pi}{q}\right) = \sum^{q-1}_{j=0} \cos\left(\frac{(2j+1) \pi}{q}\right) = 0 .$$
\end{proposition}
\begin{proof}
In each case the summands are the real parts of a set of complex numbers regularly spaced around the origin in the Argand diagram. 
The sets are $\{e^\frac{2j\,i\pi}{q}\ |\ 0\leq j<q \}$ and $\{e^\frac{(2j+1)\,i\pi}{q}\ |\ 0\leq j<q \}$. 
These symmetric sets and therefore their sums are each invariant under a rotation by $2\pi/q$ about the origin. Hence each sum is zero. 
\end{proof}
\begin{corollary}\hfill\\ \label{cor:sums}
\noindent {\rm (i)} $\quad \sum_{j=0}^{q-1} 2\sin \frac{j\pi}{q} \sin\frac{(j+1)\pi}{q} = q\cos\frac{\pi}{q}$; \\
{\rm (ii)} $\quad \sum_{j=0}^q 2\sin^2\frac{j\pi}{q} = q$.
\end{corollary}
\begin{proof}

For (i),
\begin{eqnarray*}
\sum_{j=0}^{q-1} 2\sin \frac{j\pi}{q} \sin\frac{(j+1)\pi}{q} 
  &=& \sum_{j=0}^{q-1}\left( \cos\left(\frac{j\pi}{q}-\frac{(j+1)\pi}{q}\right) - \cos\left(\frac{j\pi}{q}+\frac{(j+1)\pi}{q}\right)\right) \\
  &=& \sum_{j=0}^{q-1} \cos \frac{\pi}{q} - \sum_{j=0}^{q-1} \cos\frac{(2j+1)\pi}{q} \\
  &=& q\cos \frac{\pi}{q} - 0, \mathrm{\quad from\ Proposition~\ref{prop:sumcos}.\ }\\  
\end{eqnarray*}

For (ii),
\begin{eqnarray*}
\sum_{j=0}^{q-1} 2\sin^2 \frac{j\pi}{q} 
  &=& \sum_{j=0}^{q-1}\left(1- \cos\frac{2j\pi}{q}\right)  \\
  &=& q, \mathrm{\quad again\ from\ Proposition~\ref{prop:sumcos}.\ }\\  
\end{eqnarray*}
\end{proof}

Finally, in the analysis of our construction
the following proposition will be useful.

\begin{proposition}
\label{prop:sqrt}
Let $y$ and $z$ satisfy $y \ge 0$, $y + z \ge 0$,
and $z = O(r^{2 d})$
for some $d > 0$ as $r \to 0$.
Then
\begin{equation*}
\sqrt{y + z} = \sqrt{y} + O(r^d).
\end{equation*}
\end{proposition}

\begin{proof}
Assume without loss of generality that $z \ge 0$; then
\begin{equation*}
(\sqrt{y + z}-\sqrt{y})^2\leq (\sqrt{y + z}-\sqrt{y})(\sqrt{y + z}+\sqrt{y}) = z ,
\end{equation*}
and therefore 
$
\sqrt{y + z}-\sqrt{y} \leq \sqrt{z} = O(r^d).
$
\end{proof}

\subsection{Analysis of the construction}\label{sec:anal}

In our constructions for any fixed jump distance $\jp$, we will
take the cap and cup radius $r$ to be sufficiently small. We can
assume that $r$ is much smaller than $\jp$ and formally we will
take $r\to 0$.

For each of our lemmas we need to analyse the difference between jumps
on a hemisphere and jumps on our lawn $L$. Denote by $A$, $U$ and $S$,
the c{\bf a}ps, the c{\bf u}ps and the {\bf s}outhern hemisphere
respectively. Note that $L$ consists of $S$ with $A$ added and $U$
taken away. The set of successful jumps from $L$ to $L$ can be
classified as jumps from $S$ to $S$ plus jumps from $A$ to $S$ and
vice versa, plus jumps from $A$ to $A$, but minus jumps to or from
$U$. These latter are jumps from $U$ to $S\setminus U$ and vice versa
since our construction ensures that no jump is possible between $A$
and $U$. We account for jumps involving $U$ as jumps $U$ to $S$ plus
jumps $S$ to $U$ minus jumps $U$ to $U$, since these last were counted
twice. In symbols we may express this as:
\begin{equation}
\br{LL}\ =\ \br{SS} + \br{AS} + \br{SA} + \br{AA} - (\br{US} + \br{SU} - \br{UU}) . \label{eq:LL}
\end{equation}
For subsets $X,Y$ of the spherical surface, we denote by $\br{XY}$ the probability that 
the grasshopper starts at a point in $X$ and ends up at a point in $Y$.
For a sequence of caps corresponding to distance $\jp$ 
and a corresponding antipodal sequence, we define the following quantities:
$\br{aa}$ is the probability
of a jump from one particular cap to another particular cap distance $\jp$ from
the first, $\br{uu}$ for the corresponding probability for cups, $\br{aS}$ for
the jump probability between one cap and $S$, $\br{uS}$ for the corresponding
probability for a cup, $\br{SS}$ for the probability of a jump from $S$ to $S$,
and finally $\br{aN}$ for a jump from a cap to the northern hemisphere. We
define $\br{Sa}$ and $\br{Su}$ similarly. For Lemmas~\ref{l:peven} and~\ref{l:irrat}, each cap 
and cup is the same size. Later, for Lemma~\ref{l:podd}, this will not be the case and 
we will revise our notation accordingly.

It is easy to show that, for all $X,Y$, 
\begin{equation}\label{eq:sym1}
\br{XY} = \br{YX} .
\end{equation}
By symmetry we find that 
\begin{equation}
\br{aa} = \br{uu} \text{\  and\ } \br{uS} = \br{aN} . 
\label{eq:sym2}
\end{equation}

\begin{figure}[htp]
\begin{center}
\includegraphics[width=0.90\textwidth]{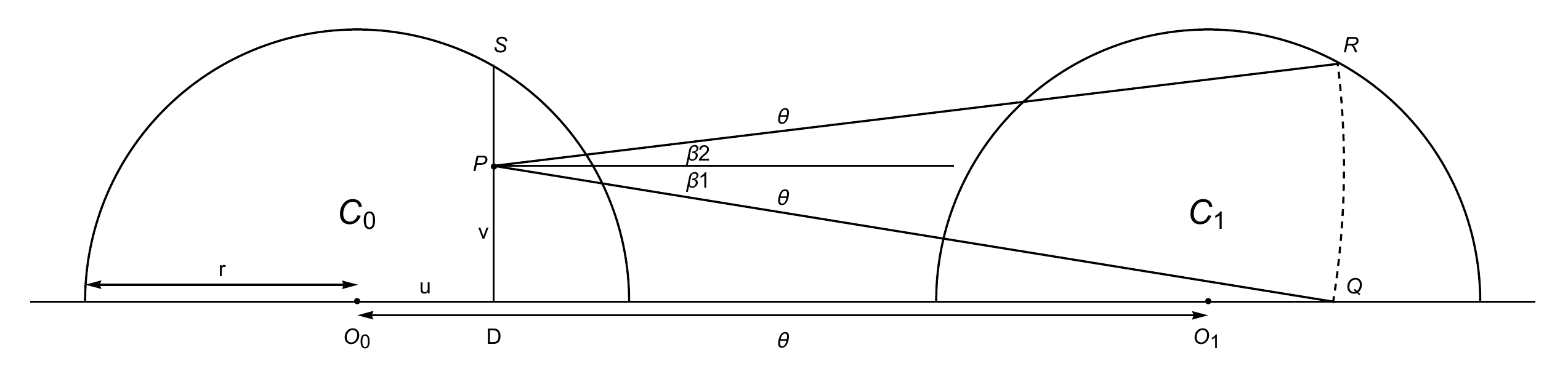}
\caption{Jump geometry\label{fig:jump}}
\end{center}
\end{figure}%

In Figure~\ref{fig:jump}, we show two successive caps $C_0,C_1$ of radius $r$,
whose centres $O_0$, $O_1$ are at angular distance $\jp$ from each other,
and a sample point $P$ in $C_0$. For jumps from $P$ towards $C_1$, $\beta_1$ and $\beta_2$ 
are the angles between the latitude through $P$ and the direction of jumps to the equator at $Q$ 
and to the circumference of $C_1$ at $R$, respectively. We see that $\beta_1\ge 0$ always, 
but it is possible that $\beta_2<0$, for example, if $P$ is close enough to the point $S$ 
in Figure~\ref{fig:jump}.

\begin{proposition}\label{prop:R}
The circle of possible jump destinations from $P$ intersects the upper semicircle 
of radius $r$ centred at $O_1$ exactly once, so $R$ is well-defined.
\end{proposition}
\begin{proof}

We first note that the point $Q$ always lies within the diameter of $C_1$. 
Let $x$ be the longitude of $Q$ relative to $O_1$, so $DQ=x+\jp-u$. By Equation~(\ref{eq:rtcos}) applied to 
the right-angled spherical triangle \sphericaltriangle $PQD$, $x$ is the solution to
\begin{equation*}\label{eq:x}
f(x):=\cos v \cos(x+\jp-u) - \cos\jp = 0 .
\end{equation*}
Since $f(x)$ is a decreasing function of $x$ and $f(u)=\cos v \cos\jp - \cos\jp \leq 0$, 
showing that $f(-r)\geq 0$ will establish that $-r\leq x \leq u$, and so $Q$ lies within the diameter of $C_1$.

Since $v\leq SD$ and $\cos u \cos SD =\cos r$ by Equation~(\ref{eq:rtcos}), 
\begin{eqnarray*}
\cos u\ f(-r) &=& \cos u \cos v \cos(-r+\jp-u) - \cos u \cos\jp \\
 &\geq& \cos r \cos(-r+\jp-u) - \cos u \cos\jp .
\end{eqnarray*}
The right-hand side of this equation is an increasing function of $u$ and
takes the value $\cos r \cos\jp - \cos(-r)\cos\jp = 0$ at $u=-r$.
Since $u\geq -r$, we have $f(-r)\geq 0$ as desired.

We use without proof the geometrically obvious fact that
the jump circle around $P$ intersects the equator exactly once to the right 
of $P$. This is true as long as $\jp$ is fixed and $r \to 0$.
Combined with the argument above, this fact implies that
the jump circle intersects the circumference of $C_1$, as well as
the circumference of the mirror image of $C_1$ below the equator.

It remains to see why the intersection point with the upper semicircle
(the circumference of $C_1$) is unique.
Notice that each of the two curves,
the jump circle and the upper semicircle,
lies in a single plane, and these two planes are different.
Therefore, each possible location for the point $R$
belongs to the line of intersection of these planes.
This line passes through at most two points of the circle around $O_1$.
We already know that the jump circle intersects not only the circumference
of $C_1$ but also its mirror image; therefore, there is exactly one intersection
point with the circumference itself.
\end{proof}

We see that 
\begin{eqnarray*}
\br{aa}&=&\int\limits_C (\beta_1 + \beta_2) \, ds,\\ 
\br{aS}&=&\int\limits_C (\pi-2\beta_1) \, ds,\\  
\br{aN}&=&\int\limits_C(\pi+2\beta_1) \, ds, \mathrm{\ and\ so\ } \\
\br{aN}-\br{aS}&=&4\int\limits_C \beta_1 \, ds,
\end{eqnarray*}
 where $ds$ is a surface element of a cap $C$.

To prepare estimates for these integrals we first find probabilities for a sample point $P$ in a cap. In the following proposition we use the notation given by Figure~\ref{fig:jump}, i.e., $u$ and $v$ are the azimuth and elevation (longitude and latitude) of~$P$ \emph{relative to the centre of} $C_0$ and $r$ is the radius of the cogs.

In the next proposition and everywhere below,
the constants in our $O(\cdot)$ notation depend on $\jp$, but not on $q$,
and we let $r \to 0$.

\begin{proposition}\label{prop:betas}  
For $0<\jp<\pi/2$,
\begin{eqnarray*}
(i)& \beta_1 &= v\cot\jp + O(r^3),\\
(ii)& \beta_1+\beta_2 &= \sqrt{r^2-u^2}\csc\jp + O(r^{3/2}).
\end{eqnarray*}
\end{proposition}

\begin{proof}
Let us first prove that
$\beta_1$ and $\beta_2$ are both $O(r)$.
Observe that $u,v$ are $O(r)$.
Notice that $\pi/2 - \beta_1$ is an angle in the right-angled spherical triangle $PDQ$.
Hence by equation~\eqref{eq:rttan},
$\sin \beta_1 = O(r) / \tan\jp = O(r)$, and so $\beta_1  = O(r)$.
Next, by the sine rule~\eqref{eq:sin} for \sphericaltriangle $PRQ$,
we have $\sin RQ / \sin (\beta_1 + \beta_2) = \sin \jp / \sin \surfaceangle PRQ$,
and so $\sin (\beta_1 + \beta_2) = O(r)$, since the points $R$ and $Q$
are both inside a circle of radius~$r$. Therefore, $\beta_1 + \beta_2 = O(r)$
and $\beta_2 = O(r)$.

We now prove parts~$(i)$ and $(ii)$ of the proposition.
We take the vertex $O_1$ in Figure~\ref{fig:jump} as the point with spherical coordinates $(0, 0)$. Then $P$ has spherical coordinates $(-\jp + u,v)$ and cartesian coordinates 
$(-\sin(\jp -u)\cos v, \cos(\jp -u)\cos v, \sin v)$. 
Let $Q$ be the point with spherical coordinates $(q,0)$ on the equator inside $C_1$  at distance $\jp$ from $P$. 
From equation~\eqref{eq:rttan} we derive
\begin{equation*}
\sin \beta_1 = \cos \surfaceangle QPD = \tan v \cot \jp = v \cot\jp + O(r^3) ,
\end{equation*}
which implies part~$(i)$ of the proposition.

Let $R$ be the point with spherical coordinates $(\lo,\la)$ on the circumference 
of $C_1$ at distance $\jp$ from $P$. Then
\begin{eqnarray}
P \cdot R &=& \cos\jp ,\label{eq:PR}\\
\cos r &=& \cos\lo \cos\la ,\label{eq:psiphi}
\end{eqnarray}
from equations~\eqref{eq:dist} and~\eqref{eq:rtcos}.

We observe that
$u,v,\lo,\la,q = O(r)$; note that for $q$ this is by Proposition~\ref{prop:R}.
Equation~\eqref{eq:psiphi} gives
$$r^2=\la^2+\lo^2 + O(r^4),$$
and we will now focus on equation~\eqref{eq:PR}.
We have
\begin{eqnarray*}
P&=&(-\sin(\jp-u)\cos v, \cos(\jp-u)\cos v,\sin v); \\
Q&=&(\sin q,\cos q,0); \\
R&=&(\sin \lo \cos \la,\cos\lo\cos \la,\sin \la).
\end{eqnarray*}
Equation~\eqref{eq:PR} thus becomes
\begin{eqnarray*}
P \cdot R &=& \cos \jp = -\sin(\jp-u)\cos v \sin \lo \cos \la + \cos(\jp-u)\cos v \cos\lo\cos \la + \sin v\sin \la \\
 &=& -\sin(\jp-u) \lo  + \cos(\jp-u) + O(r^2).
\end{eqnarray*}
So
$$\cos\jp = -\lo\sin\jp + \lo\cos\jp\sin u + \cos\jp + u \sin\jp + O(r^2) = (u-\lo)\sin\jp + \cos\jp + O(r^2),$$
giving $\lo = u + O(r^2)$.
By our choice of $Q$,
\begin{eqnarray*}
P \cdot Q &=& \cos\jp=-\sin(\jp-u)\cos v\sin q+\cos(\jp-u)\cos v\cos q \\
 &=& -q(\sin\jp\cos u-\cos\jp\sin u)+O(r^3)+\cos\jp+u\sin\jp + O(r^2) \\
 &=& (u-q)\sin\jp+\cos\jp+ O(r^2),
\end{eqnarray*}
and so $q=u+O(r^2)$ and $q-\lo=O(r^2)$.
From \sphericaltriangle $PQR$ we have
$$\cos(\beta_1+\beta_2)\sin^2\jp + \cos^2\jp = \cos QR = Q \cdot R = \sin q\sin \lo\cos \la + \cos q \cos\lo\cos \la.$$
Subtracting each side from $1$ and multiplying by $2$ gives
\begin{eqnarray*}
(\beta_1+\beta_2)^2\sin^2\jp &=& -2q\lo + q^2+\lo^2+\la^2 + O(r^4) \,
  = (q-\lo)^2+\la^2 + O(r^4) \\
 &=& \la^2 + O(r^4)=r^2-\lo^2 + O(r^4) 
  = r^2-u^2 + O(r^3). 
\end{eqnarray*}
Applying Proposition~\ref{prop:sqrt}, we get
$$(\beta_1+\beta_2)\sin\jp = \sqrt{r^2-u^2}+O(r^{3/2}),$$
which establishes part~$(ii)$.
\end{proof}

We can now compare the probabilities $\br{aa}$ and $\br{aN}-\br{aS}$.
\begin{proposition}\label{prop:aaNS} 
For $0<\jp<\pi/2$,
\begin{equation}
\frac{\br{aN}-\br{aS}}{2\br{aa}} = \cos\jp + O(\sqrt r) .
\end{equation}
\end{proposition}
\begin{proof}
To make comparison easier, it is convenient to combine the contributions from a pair of points. 
For any point $P$ we define its \emph{mate} $P'$. 
When $P$ has coordinates $(u,v)$ relative to the centre of $C_0$, 
$P'$ has coordinates $(u,v')=(u,SD-v)$, where $\cos SD \cos u = \cos r$ (see Figure~\ref{fig:jump} and equation~\eqref{eq:rtcos}). 
If $\beta_1'$ and $\beta_2'$ are the angles corresponding to $P'$, then 
\begin{equation*}
\beta_1'=(SD-v)\cot\jp + O(r^3) \mathrm{\ and\ } \beta_1'+\beta_2'= \beta_1+\beta_2 + O(r^{3/2}) ,
\end{equation*}
from Proposition~\ref{prop:betas}.
The equation $\cos SD \cos u = \cos r$ implies
\begin{equation*}
1 - \cos^2 SD = \frac{1}{\cos^2 u} \cdot (\cos^2 u - \cos^2 r) = r^2 - u^2 + O(r^4).
\end{equation*}
By Proposition~\ref{prop:sqrt} with $d = 2$, we see that
\begin{align*}
\sin SD &= \sqrt{r^2 - u^2} + O(r^2) \quad\text{and} \\
     SD &= \sqrt{r^2 - u^2} + O(r^2).
\end{align*}

For any integrand $\beta$,
\begin{equation*}
\int\limits_C \beta \, ds = \iint\limits_{C'} \beta \cos v\, du \, dv,
\end{equation*}
since the area of an element $ds$ is $\cos v\, du \, dv$, and where 
\begin{equation*}
C' = \{ (u, v) \mid \cos u \cos v \ge \cos r , -r \le u \le r \text{\ and\ } 0 \le v \le r\}.
\end{equation*}
Note however that $\cos v = 1-O(r^2)$.

When we combine the integral for a sample point $P$ with the integral for its mate $P'$ we get
\begin{align*}
2\br{aa} &= \int\limits_C (\beta_1 + \beta_2) \, ds  +\int\limits_C (\beta_1' + \beta_2') \, ds \\
&= \iint\limits_{C'} (\beta_1 + \beta_2)\cos v + (\beta_1' + \beta_2')\cos v'\, du \, dv  \\
&= \iint\limits_{C'} \left(2\sqrt{r^2 - u^2} \cdot \csc \jp + O(r^{3/2})\right) du \, dv  .\\
\end{align*}

Similarly,
\begin{align*}
\br{aN}-\br{aS} &= \int\limits_C 2\beta_1 \, ds  +\int\limits_C 2\beta_1'  \, ds \\
&= \iint\limits_{C'} 2\beta_1 \cos v + 2\beta_1' \cos v'\, du \, dv  \\
&= \iint\limits_{C'} \left(2\sqrt{r^2 - u^2} \cdot \cot \jp + O(r^3)\right) du \, dv  .\\
\end{align*}
As it is easy to check that $W=\iint\limits_{C'} \sqrt{r^2 - u^2}  du \, dv = \Theta(r^3)$,
meaning that $W = O(r^3)$ and $r^3 = O(W)$, we conclude that
\begin{align*}
\frac{\br{aN}-\br{aS}}{2\br{aa}} &= 
 \frac{2W \cdot \cot \jp + O(r^5)}{2W \cdot \csc \jp + O(r^{7/2})}\\
  &= \cos \jp + O(r^{1/2}) .
  \qedhere
\end{align*}
\end{proof}

\subsection{Rational jump, even numerator} \label{sec:peven}
We have already introduced the case $\jp=p \pi/q$ with $p/q$ irreducible and $p$ even: see Figure~\ref{fig:cc} (left). 

\begin{lemma}\label{l:peven}
For a jump of size $\jp=p \pi/q$ with $p/q$ irreducible, $p$ even and $0<\jp<\pi/2$, there is an antipodal 
lawn $L$ with greater probability of a successful jump than the hemispherical lawn.
\end{lemma}

\begin{proof}
The $q$ caps are regularly spaced at intervals $\pi/q$ around the equator alternating with their $q$ antipodal cups. 
A jump from a cap cannot get to a cup and can reach one other cap in each direction, 
while a jump from a cup cannot get to a cap and can reach one other cup in each direction. 
The caps (and similarly the cups) also correspond a cyclic sequence of $q$ jumps of size $p\pi/q$ circling the equator $p/2$ times. 

\begin{table}[thp]
\begin{center}
\begin{tabular}{|c|c|c|c|}
   \hline
   &  $A$  &  $U$ &  $S$ \\
   \hline
$A$ & $2q \cdot \br{aa}$ & 0 & $q \cdot \br{aS}$ \\
   \hline
$U$ & 0 & $2q \cdot \br{uu}$  & $q \cdot \br{uS}$ \\
   \hline
$S$ & $q \cdot \br{Sa}$ & $q \cdot \br{Su}$ & $\br{SS}$\\
   \hline
\end{tabular}
\caption{Classification of jumps for Lemma~\ref{l:peven}}
\label{table:peven}
\end{center}
\end{table}%

We summarise the total probabilities shown in Equation~\eqref{eq:LL} contributing to $\br{LL}$ in Table~\ref{table:peven}.

Using equations~\eqref{eq:LL}, \eqref{eq:sym1} and~\eqref{eq:sym2}, 
we have
\begin{eqnarray*}
\br{LL}-\br{SS}&=&\br{AS}+\br{SA}+\br{AA}-\br{US}-\br{SU}+\br{UU}\\
&=& 2q \cdot \br{aS} + 2q \cdot \br{aa} - 2q \cdot \br{uS} + 2q \cdot \br{uu} \\
&=& 4q \cdot \br{aa}-2q \cdot (\br{aN}-\br{aS}) \\
&=& 4q \cdot \br{aa}\left(1-\frac{\br{aN}-\br{aS}}{2\br{aa}}\right)\\
&=& 4q \cdot \br{aa}\left(1-(\cos\jp + O(\sqrt r))\right) ,
\end{eqnarray*}
by Proposition~\ref{prop:aaNS}.
Since $\cos\jp < 1$, if the cap radius $r$ is chosen sufficiently small then $\br{LL} > \br{SS}$.
\end{proof}


\subsection{Irrational jumps} \label{sec:irrat}
The lawn we construct for irrational $\jp/\pi$ is similar to that in the previous section but with a significant difference. 
This is previewed in Figure~\ref{fig:cc} (right).

\begin{lemma}\label{l:irrat}
For a jump of size $\jp$ where $\jp/\pi$ is irrational and $0<\jp<\pi/2$, there is an antipodal 
lawn $L$ with greater probability of a successful jump than the hemispherical lawn.
\end{lemma}

\begin{proof}
Beginning at an arbitrary point on the equator, we make a sequence of $n$ caps of radius $r$
corresponding to $n-1$ jumps of size around the equator. Cups are placed at the antipodal positions. 
The value of $n$ will be chosen later, but whatever the value of $n$, the irrationality of $\jp/\pi$ 
ensures that the centres of the caps and cups are all distinct. The value of $r$ is chosen small enough that there is no overlap 
among the caps and cups, and, further, a jump from a cap cannot get to a cup and a jump from a cup cannot get to a cap. 
The difference from the previous case is that now the forward jump from the $n$th cap in the sequence  
and the backward jump from the first cap may not reach a cap, and similarly for the cups. 
This difference is shown in Table~\ref{table:irrat}.

\begin{table}[thp]
\begin{center}
\begin{tabular}{|c|c|c|c|}
   \hline
   &  $A$  &  $U$ &  $S$ \\
   \hline
$A$ & $2(n-1) \cdot \br{aa}$ & 0 & $n \cdot \br{aS}$ \\
   \hline
$U$ & 0 & $2(n-1) \cdot \br{uu}$  & $n \cdot \br{uS}$ \\
   \hline
$S$ & $n \cdot \br{Sa}$ & $n \cdot \br{Su}$ & $\br{SS}$\\
   \hline
\end{tabular}
\caption{Classification of jumps for Lemma~\ref{l:irrat}}
\label{table:irrat}
\end{center}
\end{table}%

We now find that the total probabilities yield
\begin{eqnarray*}
\br{LL}-\br{SS}&=& 4(n-1) \cdot \br{aa}-2n \cdot (\br{aN}-\br{aS}) \\
&=& 4(n-1) \cdot \br{aa}\left(1-\left(\frac{n}{n-1}\cos\jp + O(\sqrt r)\right)\right) ,
\end{eqnarray*}
by Proposition~\ref{prop:aaNS} as before.
Since $\cos\jp < 1$, we may choose $n>1/(1-\cos\jp)$ and $r$ sufficiently small so that $\br{LL} > \br{SS}$.
\end{proof}


\subsection{Rational jumps, odd numerator} \label{sec:podd}

This final case for our main result ($\jp=p \pi/q$ with $p/q$ irreducible, $p>1$ and $p$ odd) is more complicated than the previous cases. After $q$ jumps round 
the equator, we have reached a point making an angle $p\pi$ with the initial point. 
Since $p$ is odd, this is the point antipodal to the initial point. 

\begin{lemma}\label{l:podd}
For a jump of size $\jp=p \pi/q$ with $p/q$ irreducible, $1<p<q/2$ and $p$ odd, there is an antipodal 
lawn $L$ with greater probability of a successful jump than the hemispherical lawn.
\end{lemma}

\begin{proof}
Consider the set of $2q$ points around the equator at angles $j\cdot p\pi/q$ from some initial point, for $0\leq j < 2q$. 
We will put caps at the points corresponding to $0<j<q$ and cups at points corresponding to $q<j<2q$. The 
positions for $j=0$ and $j=q$ have neither. (The arrangement is similar to the lawn in the proof of Lemma~\ref{l:irrat},
except in the present case we cannot take $n$, the number of caps, to be arbitrarily large. 
We are restricted to $n\leq q-1$, and we may have $(q-1)/(q-2) \cdot \cos (p\pi/q) \ge 1$ for large $q$.)

To achieve the required inequality we need to modify the shapes of the caps and cups. Recall that the standard caps and cups 
are bounded by semicircles of angular radius $r$. For $0\leq s\leq 1$ we define a cap $a_s$ to be just like a standard cap 
except that its upper boundary is the usual semicircle but with its height (latitude) reduced by a factor $s$:
that is, it is the set of all points $(\lo, s \la)$ such that $(\lo, \la)$ belongs to the original cap~$a$.
Note that its width (major axis) remains as $2r$. 
The analysis of jump probabilities is almost the same as in Section~\ref{sec:anal} with just obvious changes. 
Suppose we are considering jumps from a cap $a_s$ to a cap $a_t$.
Proposition~\ref{prop:R} works just as before. The results in Proposition~\ref{prop:betas} hold with the modified 
equations:
\begin{eqnarray*}
(i')& \beta_1 &= s v\cot\jp + O(r^3),\\
(ii')& \beta_1+\beta_2 &= t\sqrt{r^2-u^2}\csc\jp + O(r^{3/2}).
\end{eqnarray*}
For the integrals in the proof of Proposition~\ref{prop:aaNS}, we now derive:
$$2\br{a_s a_t} = \iint\limits_{C_r'} \left(2t\sqrt{r^2 - u^2} \cdot \csc \jp + O(r^{3/2})\right) du \, dv ,$$
and
$$\br{a_s N}-\br{a_s S} = \iint\limits_{C_r'} \left(2 s \sqrt{r^2 - u^2} \cdot \cot \jp + O(r^3)\right) du \, dv  ,$$
where 
\begin{equation*}
C_r' = \{ (u, v) \mid \cos u \cos (v/s) \ge \cos r , -r \le u \le r \text{\ and\ } 0 \le v \le s r\}
\end{equation*}
denotes a cap $a_s$.
Then 
\begin{equation}
\br{a_s a_t} = s t \cdot \br{aa} + O(r^{7/2}) \mathrm{\ and\ } \br{a_s N}-\br{a_s S} = s^2 (\br{aN}-\br{aS}) + O(r^5). 
\label{eq:asat}
\end{equation}

With preparations complete, we now define our lawn $L$. We put modified caps centred at the points round 
the equator with angles $j \jp = j p \pi/q$ for $0\leq j\leq q$. The cap corresponding to $j$ has contraction $s$ 
where $s=\sin ( j \pi /q)$. We place modified cups at points with angles $j \jp$ for $q\leq j\leq 2q$ to give an 
antipodal lawn. Note that for $j=0$, $q$ and $2q$, the contraction ratio is $0$, 
since $\sin 0=\sin \pi=\sin 2\pi=0$. 
So there is no contradiction and these two positions effectively have no cap or cup. 

With our defined sequence of caps, from equation~\eqref{eq:asat} and Corollary~\ref{cor:sums}, we see that
\begin{eqnarray*}
\br{AA} &=& \sum_{j=0}^{q-1} 2\sin \frac{j\pi}{q} \sin\frac{(j+1)\pi}{q} \cdot \br{aa} + O(r^{7/2}) \\
 &=& q\cos\frac{\pi}{q} \cdot \br{aa} + O(r^{7/2}) , \mathrm{\ and\ } \\
2(\br{AN} -\br{AS}) &=& \sum_{j=0}^q 2\sin^2\frac{j\pi}{q}  \cdot (\br{aN}-\br{aS}) + O(r^5) \\
 &=& q \cdot (\br{aN}-\br{aS}) + O(r^5) . 
\end{eqnarray*}
Now, 
\begin{eqnarray*}
\br{LL}-\br{SS} &=& 2\br{AA} - 2(\br{AN} -\br{AS}) \\
 &=& q\cos\frac{\pi}{q} \cdot 2\br{aa} - q (\br{aN}-\br{aS}) + O(r^{7/2}) \\
 &=& 2q\cdot\br{aa} \left(\cos\frac{\pi}{q} - \frac{\br{aN}-\br{aS}}{2\br{aa}} + O(r^{1/2})\right) \mathrm{\ since\ } \br{aa} = \Theta(r^3)\\
 &=& 2q\cdot\br{aa} \cdot \left(\cos\frac{\pi}{q} - \cos\jp + O(r^{1/2})\right) \\
 &>& 0
\end{eqnarray*}
for sufficiently small $r$, since $\cos\jp=\cos\frac{p\pi}{q} < \cos\frac{\pi}{q}$ for $1<p<q/2$.
\end{proof}

\subsection{Summary}
Taken together, Theorem~\ref{thm:opt} and Lemmas~\ref{l:peven}, \ref{l:irrat} and~\ref{l:podd} give our principal result.
\begin{maintheorem}\label{thm:main}
For jumps of size $\jp$ with $0<\jp<\pi/2$, the hemispherical lawn gives  
the greatest probability of a successful jump if and only if $\jp = \pi/q$ for some integer $q>1$.
\end{maintheorem}


\section{Discussion}

In quantum physics, we can represent projective measurements on a qubit
(for example, the polarization state of a photon) by an
ordered pair of antipodal points on the Bloch sphere, where
the first point represents the projector corresponding to
the outcome obtained and the second point the orthogonal
projector, which defines the other possible outcome.   
The circle represents a subset of these measurements; 
linear polarization measurements are a natural choice in
the photon case.   

One of the motivations of Refs. \cite{kent2014bloch,goulko2017grasshopper} was
to work towards a more general class of Bell inequalities,
by identifying the maximum average anti-correlation obtainable
from local hidden variable models for pairs of measurements
chosen randomly subject to the constraint that they are
separated by a fixed angle on the Bloch sphere. 
This gives a Bell inequality whenever the quantum anti-correlation for the singlet
(or any other given state) is larger.
One can also ask whether this type of ``grasshopper'' Bell inequality
can be obtained for measurements restricted to a circle on
the Bloch sphere.   The ease of linear polarization 
measurements, and the desirability of distinguishing
entangled quantum states from separable states (or
classical systems designed to try to mimic quantum states)
as efficiently as possible, make this a potentially
practically significant question.  

As noted in Ref. \cite{goulko2017grasshopper}, it is simple to find a general
analytic solution for the grasshopper problem on the real line for any
jump distance.  In contrast, solutions to the planar version 
exhibit quite complex behaviour, and analytic solutions are not
presently known.   Intuitively, this appears to reflect the
fact that the problem simplifies when the jump (which is
effectively one dimensional) takes place in a space of
the same dimension.    Our results on the circle support
this intuition, but nonetheless have non-trivial features
arising from the circle's topology. In particular, they
distinguish between a measure zero set of jumps that are rational multiples
of $\pi$ of a particular form, for which the upper bound on  lawn retention probabilities is below $1$, and the remainder, for which 
the supremum retention probabilities are $1$.

\subsection{Results on the circle}

The results for antipodal lawns are particularly interesting.
For a single antipodal lawn, they sharply 
distinguish between the case where the
jump is a rational multiple of $ \pi$ with odd numerator 
and other jump values.   In particular, they show that, even for
the single lawn grasshopper problem, informal arguments based on continuous
dependence on the jump value can fail.   
For a pair of antipodal lawns, they similarly sharply distinguish
the case of jump $\pi \frac{p}{q}$ with $p$ odd and $q$ even  
from other values.     

An antipodal lawn on the circle then represents a simple
local hidden variable model for photon polarization
measurements.   For two photons, the most general model
uses two different antipodal lawns, which are not necessarily
related.   We can motivate the simpler model given by a 
single antipodal lawn as a specific type of local hidden variable
model that attempts to reproduce the correlations 
of a singlet photon state, modelling one photon (A) by a lawn $L$ and the other (B)
by the ``opposite'' lawn $\overline{L}$.   This model 
implies that if the same linear polarization measurement is
made on both A and B then different results will always
be obtained, reproducing the perfect anti-correlation
for identical measurements exhibited by the singlet.    

In either case, Bell's theorem \cite{bell1989einstein}
shows that no local hidden variable model
can reproduce all the quantum correlations of the singlet.
Bell's theorem applies even if we restrict to linear polarization
measurements, as can be seen from 
Bell inequalities such as the CHSH \cite{clauser1969proposed} and
Braunstein-Caves \cite{braunstein1990wringing} inequalities.
These give quantitative bounds  
on sums of local hidden variable correlations, which
are violated by the (experimentally verified)
predictions of quantum theory for the singlet.

Our results show that ``grasshopper" Bell inequalities
corresponding to jumps on the circle, based
on the correlations attainable by linear photon polarization
measurements about randomly chosen axes separated by a given
fixed angle, do exist.  
However, they do so only in the case where the measurements
correspond to Bloch sphere axes separated by $\pi \frac{p}{q}$ with $p$ odd
and $q$ even. 
In these cases, the grasshopper Bell inequalities (which demonstrate
a non-trivial bound on the optimal lawn) follow from the
CHSH and Braunstein-Caves inequalities.

The CHSH and Braunstein-Caves inequalities apply to 
measurements chosen randomly from fixed finite sets.
They are robust against measurement errors.
This is because the inequalities apply to any sets containing
the right number of measurements, whether or not 
the measurements are separated by the optimal angles. 
Quantum correlations violate the inequalities maximally
when the optimal angles are chosen, but also violate
the inequalities if there are  small imprecisions in specifying
the measurements.    

It is important to stress that our generalized Bell inequalities corresponding to jumps on the circle are
{\it not} robust in the same sense.
If two parties aim to carry out 
measurements on the relevant circle that are randomly
chosen subject to the constraint that they 
are separated by angle $\pi \frac{p}{q}$ (for $p$ odd and $q$ even), but only
approximate this separation, they cannot be guaranteed
that their measurement results are inconsistent with
a local hidden variable model.   Indeed, an adversary who
knows the precise value of the separation $x \neq \pi \frac{p}{q}$ 
could substitute a pair of classical devices for their
quantum systems and achieve a {\it stronger} anticorrelation
than quantum theory predicts.   This is true no matter
how closely $x$ approximates $\pi \frac{p}{q}$, so long
as $x$ is not itself a rational multiple of $\pi$.    
Knowing a significantly better rational approximation to $x$ than
$\frac{p}{q}$ will also often suffice.   

In other words, the grasshopper Bell inequalities corresponding to 
measurements defined by a Bloch great circle display a finite precision
loophole somewhat analogous to that arising for Kochen-Specker inequalities \cite{meyer1999finite,kent1999noncontextual,clifton2000simulating,barrett2004non}.
They are mathematical results that cannot directly be tested experimentally, without
further assumptions, since in 
practice no experiment can specify the separation between measurement axes with 
the infinite precision required.\footnote{For discussions of the status of experiments
motivated by Kochen-Specker inequalities see e.g. Ref. \cite{sep-kochen-specker} and references
therein.}

That said, it is questionable how exploitable this finite precision loophole
would be in practice.   One might imagine a test of the Bell
inequalities in which pairs of measurement angles of the form
$(\cjump_A , \cjump_B)$ with $ | \cjump_B - \cjump_A | = \cjump$ 
are randomly generated in advance, and securely distributed to 
spacelike separated sites $A$ and $B$, where dials on the measurement
devices are set to measure angles $\cjump_A$ and $\cjump_B$
respectively.  
The devices might have small and consistent errors, so that settings
$\cjump_A$, $\cjump_B$ produce measurements about the axes of angles 
$\cjump_A + \delta_A$, $\cjump_B + \delta_B$ respectively.
The separations are thus actually $\cjump \pm (\delta_B - \delta_A )$,
when $\cjump_B - \cjump_A = \pm \cjump$.   
It is plausible that in this scenario the adversary (unlike
the experimenter) might know the values of $\delta_A$ and 
$\delta_B$ very precisely, perhaps for example because the adversary supplied the
experimental equipment, and plausible also that $\cjump \pm (\delta_B - \delta_A )$
may be irrational.     
If, for some reason, the experiment is restricted to one sign
choice (a one-directional grasshopper jump, say clockwise), then the 
adversary could indeed exploit the loophole, since all
jumps are through the same angle $\jp + (\delta_B - \delta_A )$,
known to the adversary but not the experimenter.

Alternatively, if, for some reason, the adversary learns in advance
which runs of the experiment involve clockwise jumps and which
involve counterclockwise (perhaps because the experimenter unwisely
carries out the first set in the morning and the second in
the afternoon, for example), she can exploit her knowledge
of the jump values in each run.   
However, without such restriction or advance knowledge, the
adversary's advantage is less clear.  Finer analysis is needed
to resolve this.  
  
In any case, the finite precision loophole on the 
circle does not imply that ``grasshopper'' Bell inequalities
obtainable on the Bloch sphere are similarly unrobust.
Although the results for the circle are relevant
to the spherical case, one should not expect
a close similarity.  The two dimensionality
of the sphere  makes a crucial difference, and 
non-zero separations between local hidden variable and quantum
anti-correlations are already known for a continuous range of jump
angles on the Bloch sphere \cite{kent2014bloch}.  
Indeed, as we discuss below, 
our results imply not only that there are new grasshopper Bell inequalities on the sphere, 
but also that for generic jumps these do not follow from known
inequalities, unlike those on the circle.

Our results do, however, emphasize the need to specify carefully
the class of measurements in any analysis, and the potential dangers
in naively assuming that it is safe to impose 
even apparently natural and
inconsequential restrictions on that class.   

Finally, we note that the circle is a natural setting
for exploring intuitions about generalisations of the grasshopper
problem such as those mentioned in Ref. \cite{goulko2017grasshopper}. 
For example, the ``ant problem'' \cite{goulko2017grasshopper}, in
which the ant attempts to walk along the path that a grasshopper
would have jumped, but dies if it ever leaves the lawn, is 
trivially solved on the circle.
This is because, for the ant, a connected lawn can never
be inferior to a disconnected lawn of the same length and is superior
unless the retention probability is zero.  
The optimal solutions are thus continuous arcs of length $L$; in the case of
antipodal lawns they are semi-circles of length $\pi$. 
Ant-optimal lawns are thus not 
generally grasshopper-optimal.    
However, if we restrict to antipodal lawns, there is a discrete set of jump 
angles for which ant-optimal lawns are also grasshopper-optimal. 

It would be interesting to explore the ``grasshopper in a breeze''
\cite{goulko2017grasshopper} on the circle, where (if we take the 
breeze to have constant angular momentum) it translates to 
requiring that the grasshopper jumps equiprobably through fixed angle $\cjump_+$ 
clockwise or fixed angle $\cjump_-$ anti-clockwise, where $\cjump_+ \neq
\cjump_-$.
It would also be interesting to explore versions \cite{goulko2017grasshopper} of the problem
in which the jump angle is drawn from a probability 
distribution (which may be symmetric about zero or, as in
the breeze example, may be asymmetric).

\subsection{Results on the sphere}

The grasshopper problem was originally motivated \cite{kent2014bloch} 
as a problem on the Bloch sphere, where antipodal lawns represent hidden variable
models for projective measurements, and this remains the most 
interesting setting for fundamental physics applications.
Although the versions of the 
problem considered previously on the plane \cite{goulko2017grasshopper}
and the versions on the circle considered above are independently
interesting, a major motivation for studying them has been to develop
intuitions for the problem on the sphere, in order ultimately to
resolve the questions about Bell inequalities raised in Ref. \cite{kent2014bloch}.

Our results represent significant progress in this direction.  
In particular, they refute the weak hemispherical colouring maximality
hypothesis set out in Ref. \cite{kent2014bloch} and so also the strong
hemispherical colouring maximality hypothesis set out in the same
paper.\footnote{The authors of Ref. \cite{kent2014bloch} carefully 
considered whether to describe either hypothesis as a conjecture
and chose not to, on the grounds that they did not have 
strong enough evidence to justify a conjecture, 
in the sense of the term commonly used by pure mathematicians.}
Theorem \ref{thm:main} shows that there are local hidden variable models for entangled two qubit 
systems that achieve stronger anticorrelations than the 
simple models defined by two hemispherical colourings for random 
pairs of measurement axes separated by fixed angle $\jp$,
for almost all angles $\jp$.   
This is intriguing, given that such models nonetheless cannot \cite{kent2014bloch}
achieve anticorrelations as strong as those predicted by 
quantum theory, at least for $0 < \jp < \frac{\pi}{3}$.  

Theorem \ref{thm:main} thus shows that there are previously 
undiscovered types of Bell inequality, separating
the predictions of local hidden variable theories and quantum theory,
that are based solely on the strength of these anticorrelations.
To identify tight versions of these Bell inequalities, one would need
to identify the supremum of the retention probabilities for pairs of
antipodal lawns, as a function of $\jp$.   
Numerical investigations \cite{vanbreugel,piosevan}, adapting the methods of 
Ref. \cite{goulko2017grasshopper}, indicate that the optimal shapes for
generic jumps up to $\approx 0.44 \pi$
resemble a cogwheel-type configuration, similar to those found in the 
planar case, and give estimates of this supremum function.  
We expect that carefully cross-checked and tested numerical results
for the spherical problem will appear in due course.   
Analytic results and proofs would of course be even more 
satisfying.   We hope our results
will motivate further work in these directions.

While our results apply directly only to Bell inequalities associated with
projective measurements on qubits, they strongly suggest that similar
types of Bell inequalities are likely to arise for other types of
measurements and for higher dimensional Hilbert spaces.   This looks
a fruitful direction to explore for extending our mathematical
understanding of quantum Bell non-locality.


\section*{Acknowledgements}

D.C. and M.P. are grateful to Sabine Hossenfelder for
the introduction to the grasshopper problem in her blog~\cite{h18}.
O.G. and A.K. are grateful to Carlo Piosevan and
Dami{\'a}n Pital{\'u}a-Garc{\'\i}a for interesting and helpful discussions.
A.K. was supported by UK Quantum Communications Hub grants
no. EP/M013472/1 and EP/T001011/1 and by an FQXi grant. 
A.K. and O.G. were supported by Perimeter Institute
for Theoretical Physics. Research at Perimeter Institute is supported
by the Government of Canada through Industry Canada and by the
Province of Ontario through the Ministry of Research and Innovation. 

\section*{References}

\bibliographystyle{unsrtnat}
\bibliography{globehopping}
\end{document}